\documentclass[11pt,a4]{article}
\usepackage[utf8]{inputenc}
\usepackage[T1]{fontenc}
\usepackage{amsmath,amsthm,amssymb}
\usepackage[ruled]{algorithm2e}
\usepackage{booktabs}
\usepackage{hyperref}
\usepackage{cleveref}
\usepackage{fullpage}
\usepackage{amssymb}
\usepackage{mathtools}
\usepackage{xcolor}
\allowdisplaybreaks
\usepackage{braket}
\usepackage{tikz}
\usetikzlibrary{positioning,shapes,calc,shapes.geometric,intersections}
\usepackage{enumitem}
\usepackage[top=25truemm,bottom=20truemm,left=20truemm,right=20truemm]{geometry}

\newtheorem{theorem}{Theorem}
\theoremstyle{plain}
\newtheorem{lemma}[theorem]{Lemma}
\newtheorem{corollary}[theorem]{Corollary}

\theoremstyle{definition}
\newtheorem{definition}[theorem]{Definition}
\newtheorem{remark}[theorem]{Remark}

\newtheorem{claim}{Claim}

\title{Graphical One-Sided Markets with Exchange Costs}

\author{
	Takashi Ishizuka\\
	Artificial Intelligence Laboratory, Fujitsu Limited, Japan\\
	\texttt{ishizuka-t@fujitsu.com}
}

\begin{document}

\maketitle

\begin{abstract}
This paper proposes a new one-sided matching market model in which every agent has a cost function that is allowed to take a negative value. Our model aims to capture the situation where some agents can profit by exchanging their obtained goods with other agents. We formulate such a model based on a graphical one-sided matching market, introduced by Massand and Simon \cite{MS19}.
We investigate the existence of stable outcomes for such a market. We prove that there is an instance that has no core-stable allocation. On the other hand, we guarantee the existence of two-stable allocations even where exchange costs exist.
However, it is $\mathtt{PLS}$-hard to find a two-stable allocation for a market with exchange costs even if the maximum degree of the graph is five.
\end{abstract}

\section{Introduction}
\label{Section_Introduction}
Allocation problems have been well-studied in economics and computer sciences.
It is known that various related problems in a wide range of fields: school choice problems \cite{APR05, APRS05, AT03}, course allocation problems \cite{Bud11, BC12, SU10}, kidney exchange problems \cite{RSU04, RSU05, RSU07}, and so on.
A stable allocation of indivisible items to rational agents is a solution concept of allocation problems; it captures the feature that any agent cannot strictly increase their utility by exchanging the assigned items with other agents.
Essentially, a standard allocation problem supposes that every agent has a preference list, which is a (strictly) total order over all items, and any transaction of items among agents after allocation yields no additional reward.

This paper focuses on the agents' behavior after allocation.
We are interested in cases where some agents profit by exchanging the assigned goods with other agents.
Specifically, we are intrigued by the behavior of agents who join a market to be rewarded by selling the obtained goods in another market despite having no preferences for goods handled in the first market.
Analyzing such an agent's behavior is motivated by real-world issues and phenomena.
Over the past three decades, the global network has grown rapidly, making it easier to conduct transactions through the network. Online platforms make it easy for anyone to sell or buy goods; one does not need to be a professional seller. In such a society, agents without preferences sometimes appear in a market for a limited number of products (such as games and plastic models, etc.) sold by lottery due to the availability of a production line.

We introduce exchange costs to deal with the above scenario.
It is natural to suppose that any exchange is realized via a mediator, and any agent needs to pay an intermediation fee to exchange her item.
We allow exchange costs to take negative values.
It seems unnatural because a negative cost means that a mediator pays the cost to an agent.
However, it is a key assumption to consider a situation where an agent benefits from an exchange.
A negative cost represents such a reward. Note that a mediator receives a cost from another agent.
Thus, a mediator will broker the exchange if the total cost is non-negative. 

This paper formulates a new one-sided matching problem where each agent has a cost function.
Our problem, {\it a graphical one-sided matching problem with exchange cost}, is based on a graphical one-sided market introduced by Massand and Simon \cite{MS19}.
We discuss the existence of stable allocations for our model. After that, we consider the computational complexity of finding a stable outcome.

\subsection{Our Contributions}
We are interested in the existence of core-stable allocation for a graphical one-sided matching problem with exchange cost. We first present that there is a graphical one-sided matching problem with exchange cost that does not have core-stable allocations. In particular, we show that there is an instance that has no $6$-stable allocation.
After that, we discuss the computational complexity of checking whether a given allocation is stable for a graphical one-sided matching problem with exchange cost. We prove that such a decision problem is $\mathtt{coNP}$-complete even without exchange cost.

We guarantee the existence of a $2$-stable allocation for a graphical one-sided matching problem with exchange cost. Our existence proof closely follows from the corresponding proof of Massand and Simon \cite{MS19}. We show that their proof also works for a problem with exchange cost.
Furthermore, we show the $\mathtt{PLS}$-completeness of finding a $2$-stable allocation for a graphical one-sided matching problem with exchange cost even if the maximum degree of the underlying player graph is at most five. 

Finally, we show that there is a polynomial-time algorithm to find a core-stable allocation for a non-graphical one-sided matching problem with exchange cost.

\subsection{Related Works}
To the best of our knowledge, we are unaware of past research on matching problems that allowed exchange costs to take negative values.

The computational aspects of matching problems have been widely studied.
For the house allocation problem, Shapley and Scarf \cite{SS74} have proposed the Top Trading Cycle algorithm, which outputs an assignment satisfying Pareto optimality, individual rationality, and strategy-proofness in polynomial time.
Hylland and Zeckhauser \cite{HZ79} have proposed a Pareto optimal, envy-free mechanism for one-sided matching markets. Unfortunately, their mechanism is not truthful.
Svensson \cite{Sve99} has proven that the Random Serial Dictatorship is the only truthful mechanism that holds ex-post Pareto optimality, anonymity, and non-bossiness.
The optimal seat arrangement problem \cite{BHJOOZ20} is very similar to the graphical one-sided matching problem. The key difference between these two models is that every agent does not have valuations for items in the optimal seat arrangement problem. Very recently, Ceylan, Chen, and Roy \cite{CCR23} have studied the computational aspects of the optimal seat arrangement problem.

A complexity class related to total search problems is useful to capture the complexity of finding a solution whose existence is guaranteed.
The complexity class $\mathtt{TFNP}$, introduced in \cite{MP91}, comprises all total search problems such that we can efficiently verify the correctness of any output solution.
Some allocation problems have been characterized in their computational complexity by subclasses of $\mathtt{TFNP}$.
Massand and Simon \cite{MS19} have proven that the problem of seeking a $2$-stable allocation for a symmetric graphical one-sided matching problem is $\mathtt{PLS}$-complete.
Cs\'{a}ji \cite{Csa22} has shown the $\mathtt{PPAD}$-completeness for computing a stable fractional matching on a hypergraphic preference system even if the maximum degree is three and the maximum edge size is three.
Recently, Goldberg, H{\o}gh, and Holledner \cite{GHH23} have shown the $\mathtt{PLS}$-completeness of computing an EFX allocation for identical agents with a submodular valuation function.

There are various scenarios of allocation problems that have local structures.
For instance, Elkind, Patel, Tsang, and Zick \cite{EPTZ20} have considered a land allocation problem where every agent has the incentive to live in a house near their friends.
In a many-to-one setting, it has been discussed that the fair allocation for each agent to a connected component of a given graph where each pair of components is disjoint \cite{BCEIP17}.

In game theory, it has also been well-studied that the computational aspects of multi-player games with a local structure.
A \textit{graphical game}, introduced by Kearns, Littman, and Sigh \cite{KLS01}, is a succinctly represented multi-player game in which players' interactions are represented by a graph, and each player's payoff depends on only actions selected by itself and its neighbors.
An \textit{identical interest polymatrix game} is well-known as a subclass of graphical games that always have pure Nash equilibria.
Cai and Daskalakis \cite{CD11} have proven the $\mathtt{PLS}$-completeness of computing a pure Nash equilibrium on an identical interest polymatrix game. 
Apt, de Keijzer, Rhan, Sch\"{a}fer, and Simon \cite{AKRSS17} investigated the existence of strong Nash equilibria in a special case of identical interest game. They have shown that there is an instance that has no strong equilibrium.

\section{Preliminaries}
We denote by $\mathbb{Z}$ and $\mathbb{R}$ the sets of integers and real numbers, respectively.
For an integer $a$ in $\mathbb{Z}$, we define $\mathbb{Z}_{\ge a} := \{ x \in \mathbb{Z} ~;~ x \ge a \}$ and $\mathbb{Z}_{> a} := \{ x \in \mathbb{Z} ~;~ x > a \}$.
Similarly, for a real number $a$ in $\mathbb{R}$, we define $\mathbb{R}_{\ge a} := \{ x \in \mathbb{R} ~;~ x \ge a \}$ and $\mathbb{R}_{> a} := \{ x \in \mathbb{R} ~;~ x > a \}$. 
We use $[n] := \{ 1, 2, \cdots n \}$ for every positive integer $n$ in $\mathbb{Z}_{> 0}$.
Let $X$ be a finite set. We denote by $|X|$ the number of elements in $X$.

\subsection{Graphical One-Sided Matching Problems with Exchange Costs}
This section presents the resource allocation problem that we study in this paper, which we call the {\it graphical one-sided matching problem with exchange cost}. Our model is based on the graphical one-sided market introduced by Massand and Simon \cite{MS19}.
As in a typical one-sided matching problem, we have a finite set of agents $N = [n]$, a finite set of items $A$, and a valuation function $v_{i}: A \to \mathbb{R}_{\ge 0}$ for each agent $i$ in $N$ in our model. We say that a valuation function $v_{i}: A \to \mathbb{R}_{\ge 0}$ is uniform if there is a non-negative real value $\hat{v} \in \mathbb{R}_{\ge 0}$ such that $v_{i}(a) = \hat{v}$ for every item $a \in A$. Throughout this paper, we assume that the numbers of agents and items are the same, i.e., $|A| = |N|$.

We suppose that items and agents in our model have local neighborhood structures, respectively.
An item graph is a simple unweighted undirected graph $\mathcal{I} = (A, \lambda)$.
An agent graph is an edge-weighted undirected graph\footnote{In the original definition in \cite{MS19}, an agent graph is defined as a simple edge-weighted digraph. This paper only focuses on the symmetric case, the undirected setting, because we are interested in the existence of core-stable allocation. Massand and Simon \cite{MS19} have shown the non-existence of $2$-stable allocation for the asymmetric setting.} $\mathcal{A} = (N, \tau, \{ w: \tau \to \mathbb{R}_{> 0} \})$, where $w$ is a function that assigns each edge $\{ i, j \}$ in $\tau$ to a positive weight $w(\{ i, j \})$ in $\mathbb{R}_{> 0}$. We say that a weight function $w: \tau \to \mathbb{R}_{> 0}$ is uniform if there exists a positive real value $\hat{w} \in \mathbb{R}_{> 0}$ such that $w(e) = \hat{w}$ for every edge $e \in \tau$.
We assume that both $\mathcal{I}$ and $\mathcal{A}$ do not have self-loops.

Each agent $i$ in $N$ has a cost function $c_{i}: A \times A \to \mathbb{R}$. Here, the value $c_{i}(a, b)$ means the cost for the agent $i$ to pay when she parts with the item $a$ and receives the item $b$.
Note that our model admits that the value $c_{i}(a, b)$ takes a negative value. We say that the cost function $c_{i}$ is trivial if for each pair of items $a, b \in A$, $c_{i}(a, b) = 0$.

An allocation $\pi: N \to A$ is a bijection that assigns agents to items. We denote by $\Pi$ the set of all allocations.
The utility function $u_{i}: \Pi \to \mathbb{R}$ of an agent $i$ in $N$ is defined by $u_{i}(\pi) = v_{i}(\pi(i)) + \sum_{ \{j, i\} \in \tau ~:~ \{ \pi(j), \pi(i) \} \in \lambda } w( \{ j, i \})$ for an allocation $\pi: N \to A$.

A coalition $X$ is a non-empty subset of agents in $N$. A coalitional exchange $(X, \mu)$ is a pair of a coalition $X \subseteq N$ and a bijection $\mu: X \to X$ satisfying that $\mu(x) \neq x$ for every $x$ in $X$.
Let $\pi: N \to A$ be an allocation.
For any coalitional exchange $(X, \mu)$, we denote by $\pi^{(X, \mu)}$ the allocation such that for each $x \in X$, $\pi^{(X, \mu)}(x) = \pi(\mu(x))$ and for each $y \in N \setminus X$, $\pi^{(X, \mu)}(y) = \pi(y)$.
We say that a coalitional exchange $(X, \mu)$ is permissible under $\pi$ if it holds that $\sum_{x \in X} c_{x}(\pi(x), \pi^{(X, \mu)}(x)) = \sum_{x \in X} c_{x}(\pi(x), \pi(\mu(x))) \ge 0$.
Given an allocation $\pi: N \to A$, we say that a permissible coalitional exchange $(X, \mu)$ under $\pi$ is blocking $\pi$ if the allocation $\pi^{(X, \pi)}$ satisfies that for every agent $x \in X$, $u_{x}(\pi^{(X, \mu)}) - c_{x}(\pi(x), \pi^{(X, \mu)}(x)) > u_{x}(\pi)$.
Let $k \ge 2$ be a positive integer.
An allocation $\pi: N \to A$ is called a $k$-stable allocation if there is no permissible coalitional exchange $(X, \mu)$ with $|X| \le k$ under $\pi$ that is blocking $\pi$.

Specifically, we call a $|N|$-stable allocation a core-stable allocation.
In the rest of this paper, a permissible coalition $(X, \mu)$ with $|X| = 2$ under an allocation $\pi$ is called a permissible pair under $\pi$. 
We simply denote by $i \leftrightarrow j$ a coalitional exchange $(\{ i, j \}, \mu)$.
Given an allocation $\pi: N \to A$ and a coalitional exchange $i \leftrightarrow j$, we denote by $\pi^{i \leftrightarrow j}$ the allocation satisfying that $\pi^{i \leftrightarrow j}(i) = \pi(j)$, $\pi^{i \leftrightarrow j}(j) = \pi(i)$, and $\pi^{i \leftrightarrow j}(k) = \pi(k)$ for every $k \in N \setminus \{ i, j \}$.
We say that a coalitional exchange $(X, \mu)$ with $|X| = 2k$ for some positive integer $k$ is a pairwise exchange if there are two non-empty subsets $X_{i} = \{ i_{1}, i_{2}, \dots, i_{k} \}$ and $X_{j} = \{ j_{1}, j_{2}, \dots, j_{k} \}$ of $X$ such that $X = X_{i} \cup X_{j}$, $X_{i} \cap X_{j} = \emptyset$, and for each $\ell \in [k]$, $\mu(i_{\ell}) = j_{\ell}$ and $\mu(j_{\ell}) = i_{\ell}$. We write this by $\{ i_{\ell} \leftrightarrow j_{\ell} \}_{\ell \in [k]}$.

\subsubsection{Computational Problems}
\label{Section_Problem_Formulation}
We formulate computational problems with respect to graphical one-sided matching problems with exchange costs.

\begin{definition}[\textsc{Exist $k$-Stable Allocation}]
	Let $k \ge 2$ be some constant integer.
	The problem \textsc{Exist $k$-Stable Allocation} is defined as follows: Given a graphical one-sided matching problem with exchange cost $\mathcal{M} = \langle \mathcal{A} = (N, \tau, \{ w: \tau \to \mathbb{Z}_{> 0} \}), \mathcal{I} = (A, \lambda), \{ v_{i}: A \to \mathbb{Z}_{\ge 0} \}_{i \in N}, \{ c_{i}: A \times A \to \mathbb{Z} \}_{i \in N} \rangle$, decide whether there exists a $k$-stable allocation $\pi: N \to A$ for $\mathcal{M}$.
\end{definition}

\begin{definition}[\textsc{Check Stable Allocation}]
	The problem \textsc{Check Stable Allocation} is defined as follows: Given a graphical one-sided matching problem with exchange cost $\mathcal{M} = \langle \mathcal{A} = (N, \tau, \{ w: \tau \to \mathbb{Z}_{> 0} \}), \mathcal{I} = (A, \lambda), \{ v_{i}: A \to \mathbb{Z}_{\ge 0} \}_{i \in N}, \{ c_{i}: A \times A \to \mathbb{Z} \}_{i \in N} \rangle$, an allocation $\pi: N \to A$, and a positive integer $k \in \{2, 3, \dots, |N| \}$, decide whether the allocation $\pi$ is a $k$-stable allocation for $\mathcal{M}$.
\end{definition}

\begin{definition}[\textsc{Check Stable Allocation without Exchange Cost}]
	The problem \textsc{Check Stable Allocation} is defined as follows: Given a graphical one-sided matching problem with exchange cost $\mathcal{M} = \langle \mathcal{A} = (N, \tau, \{ w: \tau \to \mathbb{Z}_{> 0} \}), \mathcal{I} = (A, \lambda), \{ v_{i}: A \to \mathbb{Z}_{\ge 0} \}_{i \in N} \rangle$, an allocation $\pi: N \to A$, and a positive integer $k \in \{2, 3, \dots, |N| \}$, decide whether the allocation $\pi$ is a $k$-stable allocation for $\mathcal{M}$.
\end{definition}

\begin{definition}[\textsc{Find $k$-Stable Allocation}]
	Let $k \ge 2$ be some constant integer.
	The problem \textsc{Find $k$-Stable Allocation} is defined as follows: Given a graphical one-sided matching problem with exchange cost $\mathcal{M} = \langle \mathcal{A} = (N, \tau, \{ w: \tau \to \mathbb{Z}_{> 0} \}), \mathcal{I} = (A, \lambda), \{ v_{i}: A \to \mathbb{Z}_{\ge 0} \}_{i \in N}, \{ c_{i}: A \times A \to \mathbb{Z} \}_{i \in N} \rangle$, find a $k$-stable allocation $\pi: N \to A$ for $\mathcal{M}$.
\end{definition}

\subsection{Coordination Games on Graphs}
A strategic form game $\mathcal{G} := \langle N, (S_{i})_{i \in N}, (p_{i})_{i \in N} \rangle$ consists of a set $N = [n]$ of $n \in \mathbb{Z}_{> 1}$ players, a non-empty set $S_{i}$ of pure strategies, and a payoff function $p_{i}: S_{1} \times S_{2} \times \cdots \times S_{n} \to \mathbb{R}$ for each player $i \in N$.
A strategy profile is an $n$-tuple of pure strategies $(s_{i})_{i \in N}$, where $s_{i} \in S_{i}$ for each player $i \in N$.
We denote by $\mathbb{S}_{N} := S_{1} \times S_{2} \times \cdots \times S_{n}$ the set of strategy profiles.
A coalition is a non-empty subset $X \subseteq N$ of players.

In this paper, we define a {\it network coordination game} as a subclass of identical interest polymatrix games.
A network coordination game $\mathcal{G}$ consists of an undirected graph $G = (V, E)$ without self-loop, a positive integer $m \in \mathbb{Z}_{> 0}$, and a non-empty subset of integers $S_{i} \subseteq [m]$ for each vertex $i \in V$.
Each vertex on the graph $G$ represents a player in the game $\mathcal{G}$.
For each player $i \in V$, a non-empty subset $S_{i}$ represents a pure strategy set of the player $i$.
Given a strategy profile $s = (s_{i})_{i \in V} \in \mathbb{S}_{V}$, the payoff $p_{i}(s)$ of each player $i \in V$ is defined as $| \{ \{ i, j \} \in E ~;~ s_{i} = s_{j} \}|$.

Let $\mathcal{G} = \langle G = (V, E), m,  (S_{i})_{i \in V} \rangle$ be a network coordination game, and let $k$ be a positive integer $k \in [|V|]$.
Given two strategy profiles $s = (s_{i})_{i \in V}$ and $t = (t_{i})_{i \in V}$ and a coalition $K \subseteq V$, we say that $t$ is a deviation of the players in $K$ from $s$ if $K = \{ i \in V ~;~ s_{i} \neq t_{i} \}$ and $s_{j} = t_{j}$ for every player $j \in V \setminus K$. We denote this by $s \mapsto_{K} t$. A deviation $s \mapsto_{K} t$ is profitable if $p_{i}(s) < p_{i}(t)$ for each player $i \in K$.
A strategy profile $s = (s_{i})_{i \in V}$ is called a $k$-equilibrium in $\mathcal{G}$ if there is no profitable deviation $s \mapsto_{K} t$ with the size $|K| \le k$.

\subsubsection{Computational Problems}
We formulate computational problems with respect to network coordination games.

\begin{definition}[\textsc{Exist $k$-Equilibrium}]
	Let $k$ be some constant positive integer.
	The problem \textsc{Exist $k$-Equilibrium} is defined as follows: Given a network coordination game $\mathcal{G} = \langle G = (V, E), m, (S_{i})_{i \in V} \rangle$, decide whether there is a $k$-equilibrium in $\mathcal{G}$.
\end{definition}

\begin{definition}[\textsc{Check Equilibrium}]
	The problem \textsc{Exist Equilibrium} is defined as follows: Given a network coordination game $\mathcal{G} = \langle G = (V, E), m, (S_{i})_{i \in V} \rangle$, a strategy profile $s = (s_{i})_{i \in V} \in \mathbb{S}_{V}$, and a positive integer $k \in [|V|]$, decide whether the strategy profile $s$ is a $k$-equilibrium in $\mathcal{G}$.
\end{definition}

\subsection{Complexity Classes}
We denote by $\Sigma$ a finite set of symbols. Each finite sequence of symbols in $\Sigma$ is called a string. That is, for each string $s$ with respect to $\Sigma$, there is a positive integer $n$ such that $s = s_{1} s_{2} \cdots s_{n}$, where $s_{i} \in \Sigma$ for every $i \in [n]$. We denote by $\Sigma^{*}$ a set of all finite strings, and define by $\Sigma^{n}$ a set of all strings of length $n$.

Let $R \subseteq \Sigma^{*} \times \Sigma^{*}$ be a relation.
We say that a relation $R$ is polynomially balanced if there exists a constant $c$ such that for all $(x, y)$ in $R$, $|y| \le |x|^{c}$. A relation $R$ is polynomial-time decidable if it is polynomial-time computable to decide whether a given pair of strings $(x, y) \in \Sigma^{*} \times \Sigma^{*}$ belongs to $R$.
Finally, we say that a relation $R$ is total if, for all $x \in \Sigma^{*}$, there exists at least one string $y \in \Sigma^{*}$ such that $(x, y) \in R$.

\subsubsection{Decision Problems}
Let $L \subseteq \Sigma^{*}$ be a language. 
The decision problem with respect to $L$ is defined as follows: Given a string $x \in \Sigma^{*}$, decide whether the given string $x$ is in $L$.

The complexity class $\mathtt{NP}$ is the set of all decision problems $L \subseteq \Sigma^{*}$ that there exists a polynomially balanded and polynomial-time decidable relation $R_{L} \subseteq \Sigma^{*} \times \Sigma^{*}$ such that $L = \{ x \in \Sigma^{*} ~;~ (x, y) \in R_{L} \text{ for some } y \in \Sigma^{*} \}$.
The complexity class $\mathtt{coNP}$ is a complement of $\mathtt{NP}$. Thus, a decision problem $L \subseteq \Sigma^{*}$ belongs to $\mathtt{coNP}$ if the decision problem $\Sigma^{*} \setminus L$ is in $\mathtt{NP}$.

Let $L$ and $M$ be two decision problems. A polynomial-time reduction from $L$ to $M$ is specified by a polynomial-time computable function $f: \Sigma^{*} \to \Sigma^{*}$ satisfying that for every string $x \in \Sigma^{*}$, the string $x$ is in $L$ if and only if the string $f(x)$ is in $M$.

We say that a decision problem $L$ is $\mathtt{coNP}$-hard if every decision problem in the class $\mathtt{coNP}$ is reducible to $L$ in polynomial time. We say that a decision problem $L$ is $\mathtt{coNP}$-complete if $L$ is $\mathtt{coNP}$-hard and belongs to $\mathtt{coNP}$. 

\subsubsection{Search Problems}
Let $R \subseteq \Sigma^{*} \times \Sigma^{*}$ be a relation.
The search problem with respect to $R$ is defined as follows: Given a string $x \in \Sigma^{*}$, find a string $y \in \Sigma^{*}$ such that $(x, y) \in R$ if such a string exists or report ``{\it no}'' otherwise.
We call a search problem with respect to $R$ a total search problem if $R$ is total.

The complexity class $\mathtt{TFNP}$, introduced by Megiddo and Papadimitriou \cite{MP91}, is defined as the set of all total search problems that satisfies polynomially balanced and polynomial-time decidable.

Let $R$ and $S$ be two search problems.
A polynomial-time reduction from $R$ to $S$ is specified by two polynomial-time computable functions $f: \Sigma^{*} \to \Sigma^{*}$ and $g: \Sigma^{*} \times \Sigma^{*} \to \Sigma^{*}$ satisfying that for all $x, y \in \Sigma^{*}$, if $(f(x), y) \in S$, then $(x, g(x, y)) \in R$. Thus, the function $f$ maps an instance $x$ of $R$ to an instance $f(x)$ of $S$, and the function $g$ recovers from any solution $y$ of $S$ with respect to $f(x)$ to a solution $g(x, y)$ of $R$ with respect to $x$.

Johnson, Papadimitriou, and Yannakakis \cite{JPY88} introduced the complexity class $\mathtt{PLS}$, a subclass of $\mathtt{TFNP}$, to capture the complexity of finding a local optimum in which each improvement step can be computed in polynomial time.
To define the class $\mathtt{PLS}$, we introduce \textsc{Basic PLS Problem}, and then define the class $\mathtt{PLS}$ as the set of all search problems in $\mathtt{TFNP}$ that is polynomial-time reducible to \textsc{Basic PLS Problem}.

A local search problem $\mathcal{L}$ is formulated as the following search problem: For every instance $I \in \Sigma^{*}$, there exists a finite set $F_{I} \subseteq \Sigma^{*}$ of feasible solutions, an objective function $p_{I}: F_{I} \to \mathbb{Z}_{\ge 0}$, and for every feasible solution $s \in F_{I}$, there is a neighborhood $N_{I}(s) \subseteq F_{I}$. The task of this search problem is to seek a local optimum $s^{*} \in F_{I}$ such that $p_{I}(s^{*}) \ge p_{I}(s)$ for all $s \in N_{I}(s^{*})$.

\begin{definition}
	A local search problem $\mathcal{L}$ is called the problem \textsc{Basic PLS Problem} if there exists some polynomial $q$ such that $F_{I} \subseteq \Sigma^{q(|I|)}$ for all instance $I$, and if there exists three polynomial-time algorithms $\mathtt{A}$, $\mathtt{B}$, and $\mathtt{C}$ such that:
	\begin{enumerate}[label=(\alph*)]
		\item Given an instance $I$, the algorithm $\mathtt{A}$ outputs an initial feasible solution $s_{0} \in F_{I}$.
		\item Given an instance $I$ and a string $s \in \Sigma^{q(|I|)}$, the algorithm $\mathtt{B}$ checks whether $s$ is a feasible solution, i.e., $s \in F_{I}$, and, if so, calculates the objective value $p_{I}(s)$.
		\item Given an instance $I$ and any feasible solution $s \in F_{I}$, the algorithm $\mathtt{C}$ verifies whether $s$ is a local maximum, if not, returns a feasible solution $s' \in \mathcal{N}_{I}(s)$ such that $p_{I}(s') > p_{I}(s)$.
	\end{enumerate}
\end{definition}

\begin{definition}
	The complexity class $\mathtt{PLS}$ is the set of all search problems that can reduce to \textsc{Basic PLS Problem} in polynomial time.
\end{definition}

We say that a search problem $R$ is $\mathtt{PLS}$-hard if every search problem in the class $\mathtt{PLS}$ is reducible to $R$ in polynomial time.
Since a polynomial-time reduction is transitive, a problem that can be reduced from \textsc{Basic PLS Problem} in polynomial time is a $\mathtt{PLS}$-hard problem.
We say that a search problem $R$ is $\mathtt{PLS}$-complete if $R$ is $\mathtt{PLS}$-hard and belongs to $\mathtt{PLS}$. 

\section{Main Results}
\label{Section_Results}
\subsection{Non-Existence of Core-Stable Allocations}
This section shows that there is a graphical one-sided matching problem with exchange cost that has no $6$-stable allocation. To prove this, we present a reduction from a network coordination game to a graphical one-sided matching problem with exchange cost.

Recall that Apt, de Keijzer, Rhan, Sch\"{a}fer, and Simon \cite{AKRSS17} have proven that there is a network coordination game that does not have $3$-equilibria.
Hence, a reduction from a network coordination game $\mathcal{G}$ to a graphical one-sided matching problem with exchange cost $\mathcal{M}$ such that $\mathcal{G}$ has a $k$-stable allocation if and only if $\mathcal{M}$ has a $2k$-stable allocation implies that we have a graphical one-sided matching problem with exchange cost that has no $6$-stable allocation.

\begin{theorem}[Apt, de Keijzer, Rhan, Sch\"{a}fer, and Simon \cite{AKRSS17}]
\label{Theorem_Non-Existence_AKRSS17}
	There exists a network coordination game that does not have a $3$-equilibrium.
\end{theorem}

\begin{lemma}
\label{Lemma_Reduction_from_Coordination_Games_to_One-Sided_Market}
	Let $k$ be a positive integer.
	There exists a polynomial-time reduction from \textsc{Exist $k$-Equilibria} to \textsc{Exist $2k$-Stable Allocation}.
\end{lemma}
\begin{proof}
	Here, we present a proof sketch; the formal proof can be found in Appendix \ref{Appendix__Lemma_Reduction_from_Coordination_Games_to_One-Sided_Market}.
	
	Let $\mathcal{G} := \langle G = (V, E), m, (S_{i})_{i \in V} \rangle$ be an instance of the problem \textsc{Exist $k$-Equilibria}.
	We now construct an instance $\mathcal{M} = \langle \mathcal{A} = (N, \tau, \{ w: \tau \to \mathbb{Z}_{> 0} \}), \mathcal{I} = (A, \lambda), \{ v_{i}: A \to \mathbb{Z}_{\ge 0} \}_{i \in N}, \{ c_{i}: A \times A \to \mathbb{Z} \}_{i \in N } \rangle$ of the problem \textsc{Exist $2k$-Stable Allocation} such that there is a $k$-equilibrium in $\mathcal{G}$ if and only if there is a $2k$-stable allocation for $\mathcal{M}$.
	
	First, we construct the item graph $\mathcal{I} = (A, \lambda)$. For each player $i \in V$, we define $A_{i} := \{ a_{i, s} ~;~ s \in S_{i} \}$. The set of items is defined as $A = \bigcup_{i \in V} A_{i}$. For each item $a \in A$, there is a pair of a player $i \in V$ and a strategy $s \in S_{i}$ such that $a = a_{i, s}$. We define the label function $f: A \to [m]$ by $f(a_{i, s}) = s \in S_{i}$ for each item $a_{i, s}$ in $A$.
	We set $\lambda$ to be that for each pair of two distinct items $a$ and $b$ in $A$, $\{ a, b \} \in \lambda$ if and only if $f(a) = f(b)$.
	
	Next, we construct the agent graph $\mathcal{A} = ( N, \tau, \{ w: \tau \to \mathbb{Z}_{> 0} \} )$. For each player $i \in V$, we define a partial agent set $N_{i}$ as $\{ x_{i} \} \cup \{ d_{i}^{1}, d_{i}^{2}, \dots, d_{i}^{|S_{i}| - 1} \}$. The agent set is defined as $N = \bigcup_{i \in V} N_{i}$. For every player $i$ in $V$, we call an agent in $\{ d_{i}^{1}, d_{i}^{2}, \dots, d_{i}^{|S_{i}| - 1} \}$ a dummy agent. For each player $i$ in $V$, we simply write $\{ d_{i}^{1}, d_{i}^{2}, \dots, d_{i}^{|S_{i}| - 1} \}$ by $D_{i}$. The set $\tau$ of edges is defined as $E$. Thus, every dummy agent is isolated. For each edge $e \in \tau$, we define the edge weight as $w(e) = 2$.
	
	Finally, for each agent $i$ in $N$, we define the valuation function $v_{i}: A \to \mathbb{Z}_{\ge 0}$ and the cost function $c_{i}: A \times A \to \mathbb{Z}$ to satisfy the following three properties.
	\begin{enumerate}[label=(\arabic*)]
		\item For any positive integer $k' \ge 2$, every $k'$-stable allocation $\pi: N \to A$ is valid. Here, an allocation $\pi: N \to A$ is {\it valid} if for every agent $\xi_{i}$ in $N_{i}$, it holds that $\pi(\xi_{i}) \in A_{i}$.
		\item Let $\pi: N \to A$ be a valid allocation. Every permissible coalitional exchange $(X, \mu)$ under $\pi$ is a pairwise exchange. In particular, for each permissible coalitional exchange $(X, \mu)$ under $\pi$, there exists a coalition $K$ of players in $V$ such that for each player $i \in K$, we have two agents $x_{i}$ and $d_{i} \in D_{i}$ satisfying that $(X, \mu) = \{ x_{i} \leftrightarrow d_{i} \}_{i \in K}$.
		\item Let $\pi: N \to A$ be a valid allocation. We define the strategy profile $s_{\pi} \in \mathbb{S}_{V}$ as $s_{\pi}(i) = f(\pi(x_{i}))$ for each player $i \in V$. Then, for each player $i \in V$, it satisfies that $u_{x_{i}}(\pi) = \Delta + 2 \cdot p_{i}(s_{\pi})$, where the constant $\Delta$ is defined to be sufficiently large.
	\end{enumerate}
	
	What remains is to prove that the network coordination game $\mathcal{G}$ has a $k$-equilibrium if and only if the graphical one-sided matching problem with exchange cost $\mathcal{M}$ has a $2k$-stable allocation.
	We describe the correspondence between a strategy profile and an allocation. For an allocation $\pi: N \to A$, we consider the strategy profile $s_{\pi} = (s_{\pi}(i))_{i \in V} \in \mathbb{S}_{V}$ such that $s_{\pi}(i) = f(\pi(x_{i}))$ for each player $i \in V$. If $\pi$ is $2k$-stable, then $s_{\pi}$ is a $k$-equilibrium.
	On the other hand, for a strategy profile $s = (s_{i})_{i \in V} \in \mathbb{S}_{V}$, we consider the allocation $\pi_{s}: N \to A$ such that for each player $i \in V$, $\pi_{s}(x_{i}) = a_{i, s_{i}}$, and each dummy agent $d_{i}^{k}$ in $D_{i}$ is assigned relevant item in $A_{i}$.
	To determine a relevant item, we induce a linear order over the rest items. According to that order, we assign an item to a dummy agent. See the formal proof for details.
	 If $s$ is a $k$-equilibrium, then $\pi_{s}$ is $2k$-stable.
\end{proof}

\begin{theorem}
\label{Theorem_Non-Existence_Six-Stable}
	There exists a graphical one-sided matching problem with exchange cost that does not have $6$-stable allocations.
\end{theorem}
\begin{proof}
	It immediately follows from Theorem \ref{Theorem_Non-Existence_AKRSS17} and Lemma \ref{Lemma_Reduction_from_Coordination_Games_to_One-Sided_Market}.
\end{proof}

Next, we consider the complexity of checking whether an allocation for a graphical one-sided matching problem with exchange cost is a $k$-stable allocation for some positive integer $k \ge 2$.
The following theorem states that such a problem is $\mathtt{coNP}$-complete even if every player's cost function is trivial.

\begin{theorem}
\label{Theorem_Check_Stable_without_Costs_coNP}
	The problem \textsc{Check Stable Allocation without Exchange Cost} is $\mathtt{coNP}$-complete.
\end{theorem}
\begin{proof}
	Here, we present a proof sketch; the formal proof can be found in Appendix \ref{Appendix__Theorem_Check_Stable_without_Costs_coNP}.
	
	It is easy to see that the problem \textsc{Check Stable Allocation without Exchange Cost} belongs to $\mathtt{coNP}$.
	Therefore, it is sufficient to prove the $\mathtt{coNP}$-hardness of the problem \textsc{Check Stable Allocation without Exchange Cost}.
	For proving the $\mathtt{coNP}$-hardness, we show a polynomial-time reduction from \textsc{Check Equilibrium}, which is a $\mathtt{coNP}$-complete problem \cite{AKRSS17}, to \textsc{Check Stable Allocation without Exchange Cost}.
	
	Let $\left\langle \mathcal{G} = \left( G = (V, E), m, (S_{i})_{i \in V} \right), s = (s_{i})_{i \in V}, k \right\rangle$ be an instance of \textsc{Check Equilibrium}.
	We now construct a graphical one-sided matching problem without exchange costs $\mathcal{M} = \left\langle \mathcal{A} = (N, \tau, \{ w: \tau \to \mathbb{Z}_{> 0} \}), \right. \allowbreak \left. \mathcal{I} = (A, \lambda), \{ v_{i}: A \to \mathbb{Z}_{\ge 0} \}_{i \in V} \right\rangle$ and an allocation $\pi: N \to A$.
	The construction of the agent graph $\mathcal{A}$ and the item graphs $\mathcal{I}$ is the same as the graph shown in the proof of Lemma \ref{Lemma_Reduction_from_Coordination_Games_to_One-Sided_Market}.
	Due to the different computational problems --- specifically, without exchange costs --- we modify the valuation function $v_{d}: A \to \mathbb{Z}_{\ge 0}$ for each dummy agent $d \in \bigcup_{i \in V} D_{i}$.
	For each player $i \in V$ and for each dummy agent $d_{i} \in D_{i}$, we define the new valuation function $v_{d_{i}}: A \to \mathbb{Z}_{\ge 0}$ as follows: $v_{d_{i}}(a_{i, s_{i}}) = \Delta$ and $v_{d_{i}}(a) = 0$ for every item $a \in A \setminus \{ a_{i, s_{i}} \}$.
	Finally, we define the allocation $\pi: N \to A$ corresponding to the strategy profile $s = (s_{i})_{i \in V}$ as follows.
	For each player $i \in V$, $\pi(x_{i}) = a_{i, s_{i}}$. For every dummy agent, we assign a relevant item using the same approach as in the proof of Lemma \ref{Lemma_Reduction_from_Coordination_Games_to_One-Sided_Market}.
	Then, it holds that the strategy profile $s$ is a $k$-equilibrum of $\mathcal{G}$ if and only if the allocation $\pi: N \to A$ is a $2k$-stable allocation of $\mathcal{M}$.
\end{proof}

We obtain the $\mathtt{coNP}$-completeness of the problem \textsc{Check Stable Allocation} as a corollary of Theorem \ref{Theorem_Check_Stable_without_Costs_coNP} since there is a trivial polynomial-time reduction from \textsc{Check Stable Allocation without Exchange Cost} to \textsc{Check Stable Allocation}.

\begin{corollary}
	The problem \textsc{Check Stable Allocation} is $\mathtt{coNP}$-complete.
\end{corollary}

\subsection{Existence of Two-Stable Allocations}
This section discusses the existence and the complexity of a $2$-stable allocation for a graphical one-sided matching problem with exchange cost.
We show that every such problem has at least one $2$-stable allocation.
Our existence proof closely follows from the corresponding proof of Massand and Simon \cite{MS19}.
Furthermore, we observe that the problem of finding a $2$-stable allocation is in $\mathtt{PLS}$.

The proofs of Theorem \ref{Theorem_Existence_Two-Stable} and Lemma \ref{Lemma_Find_Two_Stable_in_PLS} can be found in Appendices \ref{Appendix__Theorem_Existence_Two-Stable} and \ref{Appendix__Lemma_Find_Two_Stable_in_PLS}, respectively.

\begin{theorem}
\label{Theorem_Existence_Two-Stable}
	Every graphical one-sided matching problem with exchange costs always has at least one $2$-stable allocation.
\end{theorem}

\begin{lemma}
\label{Lemma_Find_Two_Stable_in_PLS}
	The problem \textsc{Find $2$-Stable Allocation} belongs to $\mathtt{PLS}$.
\end{lemma}

We consider the computational hardness of finding a $2$-stable allocation for a graphical one-sided matching problem with exchange cost. 
Recall that Massand and Simon \cite{MS19} have proven the $\mathtt{PLS}$-completeness of computing a $2$-stable allocation for a graphical one-sided matching problem without exchange cost even if the maximum degree of the underlying agent graph is at most six.
Thus, it also follows the $\mathtt{PLS}$-hardness for \textsc{Find $2$-Stable Allocation}.

The next lemma gives us a stronger $\mathtt{PLS}$-hardness result: In the setting where exchange cost exists, the problem of finding a $2$-stable allocation is $\mathtt{PLS}$-hard even if the maximum degree of the underlying agent graph is at most five.

\begin{lemma}
\label{Lemma_PLS-Hardness_for_Five-Degree}
	The problem \textsc{Find $2$-Stable Allocation} is $\mathtt{PLS}$-hard even if the maximum degree of the underlying agent graph is at most five.
\end{lemma}
\begin{proof}
	Here, we present a proof sketch; the formal proof can be found in Appendix \ref{Appendix__Lemma_PLS-Hardness_for_Five-Degree}.
	
	We now show a polynomial-time reduction from the problem \textsc{Local Max-Cut} to the problem \textsc{Find $2$-Stable Allocation}.
	The problem \textsc{Local Max-Cut} is a total search problem: Given an edge-weighted undirected graph $\mathcal{G} = (V, E, \{ \omega: E \to \mathbb{Z}_{> 0} \})$, find an assignment $x \in \{ 0, 1 \}^{V}$ such that for each vertex $v \in V$, $\sum_{ \{ v, u \} \in E ~:~ x(v) \neq x(u) } \omega(\{ v, u \}) \ge \sum_{ \{ v, u \} \in E ~:~ x(v) = x(u) } \omega(\{u, v\})$.
	The problem \textsc{Local Max-Cut} is one of the well-known $\mathtt{PLS}$-complete problems. Specifically, Els\"{a}sser and Tscheuschner \cite{ET11} have shown that the problem \textsc{Local Max-Cut} is $\mathtt{PLS}$-complete even if the maximum degree of the underlying graph is at most five.
	
	Our reduction is inspired by the corresponding proof by Massand and Simon\footnote{Remark that the proof in \cite{MS19} can be viewed as a polynomial-time reduction from a uniformed balanced variant of \textsc{Local Max-Cut}. Such a problem is also $\mathtt{PLS}$-complete \cite{SY91}.} \cite{MS19}.
	We also use a complete bipartite graph in which both sides have the same number of nodes with the number of agents as the item graph. Due to prove the stronger hardness results, we modify the agent graph.
	
	Let $\mathcal{G} = (V, E, \{ \omega: E \to \mathbb{Z}_{> 0} \})$ be a \textsc{Local Max-Cut} instance.
	We construct an instance $\mathcal{M}$ of a graphical one-sided matching problem with exchange cost from $\mathcal{G}$.
	First, we describe the structure of the item graph $\mathcal{I} = (A, \lambda)$. The set of items $A$ is a union of the sets of $|V|$ zero-labeled items $A_{0}$ and $|V|$ one-labeled items $A_{1}$, i.e., $A = A_{0} \cup A_{1}$. The edge set $\lambda$ of the item graph is defined as $\left\{ \{ a_{0}, a_{1} \} ~;~ \forall a_{0} \in A_{0}, \forall a_{1} \in A_{1} \right\}$. Thus, the item graph $\mathcal{I}$ is a complete bipartite graph between zero-labeled items and one-labeled items.
	Second, we describe the structure of the agent graph $\mathcal{A} = (N, \tau, \{ w: \tau \to \mathbb{Z}_{> 0} \})$. The set of agents $N$ is a union of the original vertex set $V$ and the set of $|V|$ dummy agents $D = \{ d_{1}, d_{2}, \dots, d_{|V|} \}$, i.e., $N = V \cup D$. The structure of the agent graph $\mathcal{A}$ is almost the same as the original player interaction graph $G$. The set of edges $\tau$ is defined as $E$, i.e., every dummy agent is isolated. The edge-weight function $w: \tau \to \mathbb{Z}_{> 0}$ is defined by $w(e) = 2 \cdot \omega(e)$ for each edge $e \in E$.
	Finally, we define the valuation function $v_{i}: A \to \mathbb{Z}_{\ge 0}$ and the cost function $c_{i}: A \times A \to \mathbb{Z}$ for each agent $i$ in $N$.
	Every agent has a uniform valuation function. 
	Each agent $i_{v}$ in $V$ has the cost function $c_{i_{v}}$ such that for each pair of two items $a$ and $b$ in $A$, $c_{i_{v}}(a, b) = 1$ if $a$ and $b$ belong to the different sets and $c_{i_{v}}(a, b) = 0$ otherwise.
	Each agent $i_{d}$ in $D$ has the cost function $c_{i_{d}}$ such that for each pair of two items $a$ and $b$ in $A$, $c_{i_{d}}(a, b) = -1$ if $a$ and $b$ belong to the different sets and $c_{i_{d}}(a, b) = 0$ otherwise.

	What remains is to prove that we can efficiently recover a solution of $\mathcal{G}$ from a $2$-stable allocation on $\mathcal{M}$.
	We describe how to convert an assignment from an allocation. For any allocation $\pi: N \to A$, we consider the allocation $x_{\pi} \in \{0, 1\}^{V}$ such that for each vertex $i \in V$, $x_{\pi}(i) = 0$ if and only if $\pi(i) \in A_{0}$. If $\pi: N \to A$ is a $2$-stable allocation, then the assignment $x_{\pi}$ is a solution of $\mathcal{G}$.
\end{proof}

From the above two lemmata, we immediately obtain the following theorem.

\begin{theorem}
\label{Theorem_Degree-5_Weak_Instance_PLS-complete}
	The problem \textsc{Find $2$-Stable Allocation} is $\mathtt{PLS}$-complete even if the maximum degree of the underlying agent graph is at most five.	
\end{theorem}

\subsection{Polynomial-Time Computability}
\label{Section_Polynomial-Time_Computability}
Finally, we discuss the polynomial-time computability for computing a core-stable allocation for some specific graphical one-sided matching problem with exchange cost.
We show that there is a polynomial-time algorithm to find a core-stable allocation in a graphical one-sided matching problem with exchange cost when agents and items do not comprise local structures. Thus, we can also efficiently solve a typical one-sided matching problem even if we admit exchange costs.

\begin{theorem} \label{Theorem_Polynomial-time_Computable}
	We have a polynomial-time algorithm to find a core-stable allocation when a graphical one-sided matching problem with exchange cost satisfies at least one of the following two properties:
	The underlying agent graph is either edge-empty or a uniform-weighted complete graph; and the underlying item graph is either edge-empty or complete.
\end{theorem}
\begin{proof}
	Here, we present a proof sketch; the formal proof can be found in Appendix \ref{Appendix__Theorem_Polynomial-time_Computable}.
	
	Since these problems are reducible to non-graphical ones in polynomial time, it is sufficient to show that we have a polynomial-time algorithm to find a core-stable allocation on a non-graphical one-sided matching problem with exchange cost.
	To prove this, we construct an instance of the maximum weight matching in a bipartite graph from an instance of a non-graphical one-sided matching problem with exchange cost. We can solve the maximum weight perfect matching on complete bipartite matching in polynomial time \cite{DS12}.
\end{proof}

\begin{remark}
	In the setting where exchange cost exists, the serial dictatorship mechanism does not work well.
	Consider the following non-graphical one-sided matching problem with exchange cost where two agents $i_{1}$ and $i_{2}$ exit. We denote by $a$ and $b$ two items. The valuation functions are $v_{i_{1}}(a) = 10$, $v_{i_{1}}(b) = 9$, $v_{i_{2}}(a) = 8$, and $v_{i_{2}}(b) = 5$. Hence, both agents strictly prefer $a$ to $b$.
	The cost functions are $c_{i_{1}}(a, b) = c_{i_{1}}(b, a) = -2$ and $c_{i_{2}}(a, b) = c_{i_{2}}(b, a) = 2$.
	When the serial dictatorship mechanism sequentially allocates items to agents such that $i_{1}$ has $a$ and $i_{2}$ has $b$, this allocation is not stable. The agent $i_{2}$'s utility for the item $a$ is higher than for the item $b$, even if she incurs the cost.
	Thus, depending on ordering over agents, such an algorithm fails to output a stable allocation.
\end{remark}

\section{Conclusions}
\label{Section_Conclusions}
We have investigated the computational aspects of a graphical one-sided matching problem with exchange cost, a new one-sided market model to capture the scenario where some agents profit by exchanging their assigned items with other agents.
We have proven that there is such an instance that has no $6$-stable allocation.
On the other hand, we have shown that every graphical one-sided matching problem with exchange cost always has at least one $2$-stable allocation.
However, it is $\mathtt{PLS}$-hard to find it even when the maximum degree of the underlying agent graph is at most five.

This paper leaves several open questions.
The first intriguing question is the existence of a $k$-stable allocation for a graphical one-sided matching problem with exchange cost for each integer $k$ in $\{ 3, 4, 5 \}$.
Another fascinating open question is the existence of a core-stable allocation for a graphical one-sided matching problem without exchange cost. Note that although we have proven the $\mathtt{coNP}$-completeness of checking whether a given allocation is a $k$-stable allocation for a given graphical one-sided matching problem without exchange cost; this result does not imply that the existence or the non-existence of core-stable allocations for a graphical one-sided matching problem without exchange cost.

\section*{Acknowledgment}
This work was partially supported by JST, ACT-X, Grant Number JPMJAX2101.
Also, this work included a part done while the author had been supported by JSPS KAKENHI Grant Number JP21J10845.
The author would like to thank Hiroaki Iwashita, Kentaro Kanamori, Hirofumi Suzuki, and Takuya Takagi for many useful conversations throughout the work.



\bibliographystyle{plain}
\bibliography{references}

\clearpage

\appendix

\section{Proof of Lemma \ref{Lemma_Reduction_from_Coordination_Games_to_One-Sided_Market}}
\label{Appendix__Lemma_Reduction_from_Coordination_Games_to_One-Sided_Market}
In this section, we prove Lemma \ref{Lemma_Reduction_from_Coordination_Games_to_One-Sided_Market}.
That is, we show a polynomial-time reduction from \textsc{Exist $k$-Equilibria} to \textsc{Exist $2k$-Stable Allocation}.

Let $\mathcal{G} := \langle G = (V, E), m, (S_{i})_{i \in V} \rangle$ be an instance of the problem \textsc{Exist $k$-Equilibria}.
We now construct an instance $\mathcal{M} = \langle \mathcal{A} = (N, \tau, \{ w: \tau \to \mathbb{Z}_{> 0} \}), \mathcal{I} = (A, \lambda), \{ v_{i}: A \to \mathbb{Z}_{\ge 0} \}_{i \in N}, \{ c_{i}: A \times A \to \mathbb{Z} \}_{i \in N } \rangle$ of the problem \textsc{Exist $2k$-Stable Allocation} such that there is a $k$-equilibrium in $\mathcal{G}$ if and only if there is a $2k$-stable allocation for $\mathcal{M}$.

First, we construct the item graph $\mathcal{I} = (A, \lambda)$. For each player $i \in V$, we define $A_{i} := \{ a_{i, s} ~;~ s \in S_{i} \}$. The set of items $A$ is defined by $A = \bigcup_{i \in V} A_{i}$. For each item $a \in A$, there is a pair of a player $i \in V$ and a strategy $s \in S_{i}$ such that $a = a_{i, s}$. We define the label function $f: A \to [m]$ as $f(a_{i, s}) = s \in S_{i}$ for each item $a_{i, s}$ in $A$. 
We set $\lambda$ to be that for each pair of two distinct items $a$ and $b$ in $A$, $\{ a, b \} \in \lambda$ if and only if $f(a) = f(b)$.

Next, we define the agent graph $\mathcal{A} = ( N, \tau, \{ w: \tau \to \mathbb{R}_{> 0} \} )$. For each player $i \in V$, we define a partial agent set $N_{i}$ as $\{ x_{i} \} \cup \{ d_{i}^{1}, d_{i}^{2}, \dots, d_{i}^{|S_{i}| - 1} \}$. The agent set $N$ is defined by $N = \bigcup_{i \in V} N_{i}$. For every player $i$ in $V$, we call an agent in $\{ d_{i}^{1}, d_{i}^{2}, \dots, d_{i}^{|S_{i}| - 1} \}$ a dummy agent. For each player $i$ in $V$, we simply write $\{ d_{i}^{1}, d_{i}^{2}, \dots, d_{i}^{|S_{i}| - 1} \}$ by $D_{i}$. The set $\tau$ of edges is defined as $E$. Thus, every dummy agent is isolated. For each edge $e \in \tau$, we define the edge weight by $w(e) = 2$.

Finally, for each agent $i$ in $N$, we define the valuation function $v_{i}: A \to \mathbb{Z}_{\ge 0}$ and the cost function $c_{i}: A \times A \to \mathbb{Z}$ to satisfy the following three properties.
\begin{enumerate}[label=(\arabic*)]
	\item For any positive integer $k' \ge 2$, every $k'$-stable allocation $\pi: N \to A$ is valid. \label{First_Property}
	\item Let $\pi: N \to A$ be a valid allocation. Every permissible coalitional exchange $(X, \mu)$ under $\pi$ is a pairwise exchange. In particular, for each permissible coalitional exchange $(X, \mu)$ under $\pi$, there exists a coalition $K$ of players in $V$ such that for each player $i \in K$, we have two agents $x_{i}$ and $d_{i} \in D_{i}$ satisfying that $(X, \mu) = \{ x_{i} \leftrightarrow d_{i} \}_{i \in K}$. \label{Second_Property}
	\item Let $\pi: N \to A$ be a valid allocation. We define the strategy profile $s_{\pi} \in \mathbb{S}_{V}$ as $s_{\pi}(i) = f(\pi(x_{i}))$ for each player $i \in V$. Then, for each player $i \in V$, it satisfies that $u_{x_{i}}(\pi) = \Delta + 2 \cdot p_{i}(s_{\pi})$. \label{Third_Property}
\end{enumerate}
Here, an allocation $\pi: N \to A$ is valid if for every agent $\xi_{i}$ in $N_{i}$, it holds that $\pi(\xi_{i}) \in A_{i}$, and the constant $\Delta$ is defined as $10 \cdot \max_{i \in V} | \{ \{ i, j \} \in E \} |$.
We use the above three properties to prove the correctness of our reduction.

Now, we describe how to construct the valuation functions and the cost functions.
For each agent $\xi_{i} \in N_{i}$, the valuation function $v_{i}: A \to \mathbb{Z}_{\ge 0}$ is defined to be that $v_{\xi_{i}}(a) = \Delta$ if $a \in A_{i}$ and $v_{\xi_{i}}(a) = 0$ otherwise.

Next, we describe the construction of the cost functions.
For each player $i \in V$, we define the cost function $c_{x_{i}}: A \times A \to \mathbb{Z}$ by

\begin{align*}
	c_{x_{i}}(a, b) = \begin{cases}
		0 & \text{ if } a = b,\\
		1 & \text{ if } a \neq b, a \in A_{i}, \text{ and } b \in A_{i},\\
		\frac{\Delta}{2} & \text{ if } a \neq b, a \not\in A_{i}, \text{ and } b \in A_{i},\\
		- \frac{\Delta}{2} & \text{ if } a \neq b, a \not\in A_{i}, \text{ and } b \not\in A_{i},\\
		- |V| \cdot \Delta & \text{ if } a \neq b, a \in A_{i}, \text{ and } b \not\in A_{i},
	\end{cases}
\end{align*}

\noindent
and for each agent $d_{i} \in D_{i}$, we define the cost function $c_{d_{i}}: A \times A \to \mathbb{Z}$ by

\begin{align*}
	c_{d_{i}}(a, b) = \begin{cases}
		0 & \text{ if } a = b,\\
		-1 & \text{ if } a \neq b, a \in A_{i}, \text{ and } b \in A_{i},\\
		\frac{\Delta}{2} & \text{ if } a \neq b, a \not\in A_{i}, \text{ and } b \in A_{i},\\
		- \frac{\Delta}{2} & \text{ if } a \neq b, a \not\in A_{i}, \text{ and } b \not\in A_{i},\\
		- |V| \cdot \Delta & \text{ if } a \neq b, a \in A_{i}, \text{ and } b \not\in A_{i}.
	\end{cases}
\end{align*}

We complete the construction of the graphical one-sided matching problem with exchange cost $\mathcal{M}$.
It is easy to see that our reduction can be computed in polynomial time.

From now on, we show that our construction satisfies the above three properties.
First, we prove Property (\ref{First_Property}): Every $k'$-stable allocation of $\mathcal{M}$ is valid for any positive integer $k' \ge 2$.

\begin{claim}
\label{Claim_Stable_is_Valid}
	For any positive integer $k' \ge 2$, every $k'$-stable allocation $\pi: N \to A$ is valid.
\end{claim}
\begin{proof}
	For the sake of contradiction, we suppose that there is an agent $\xi_{i} \in N_{i}$ such that $\pi(\xi_{i}) \not\in A_{i}$. Then, there exists another agent $\xi_{j} \in N_{j}$ with $j \neq i$ such that $\pi(\xi_{j}) \in A_{i}$.
	These two agents $\xi_{i}$ and $\xi_{j}$ are permissible under the allocation $\pi$ since
	\begin{align*}
		c_{\xi_{i}}(\pi(\xi_{i}), \pi(\xi_{j})) + c_{\xi_{j}}(\pi(\xi_{j}), \pi(\xi_{i})) 
			= \frac{\Delta}{2} + c_{\xi_{j}}(\pi(\xi_{j}), \pi(\xi_{i}))
			\ge \frac{\Delta}{2} - \frac{\Delta}{2} = 0.
	\end{align*}
	We now show that the above permissible pair $(\xi_{i}, \xi_{j})$ blocks the allocation $\pi$.
	First, we can easily see that the utility of the agent $\xi_{i}$ is strictly increasing by exchanging with $\xi_{j}$:
	\begin{align*}
		u_{\xi_{i}}(\pi^{\xi_{i} \leftrightarrow \xi_{j}}) - c_{\xi_{i}}(\pi(\xi_{i}), \pi(\xi_{j}))
			&= \Delta + \sum_{ \{ \xi_{i}, \eta \} \in \tau ~:~ \{ \pi^{\xi_{i} \leftrightarrow \xi_{j}}(\xi_{i}), \pi^{\xi_{i} \leftrightarrow \xi_{j}}(\eta) \} \in \lambda } w( \{ \xi_{i}, \eta \}) - \frac{\Delta}{2}\\
			&\ge \frac{\Delta}{2} = 5 \cdot \max_{i \in V} |\{ \{ i, j \} \in E \}|\\
			&> 2 \cdot |\{ \xi_{i}, \eta \} \in E \}| \ge u_{\xi_{i}}(\pi).
	\end{align*}
	Next, we prove that the utility of the agent $\xi_{j}$ is strictly increasing by exchanging with $\xi_{i}$.
	There are two cases: One is $\pi(\xi_{i}) \in A_{j}$, and the other is $\pi(\xi_{i}) \in A \setminus (A_{i} \cup A_{j})$.
	When the agent $\xi_{i}$ is allocated the item in $A_{j}$ by $\pi$, we have that
	\begin{align*}
		u_{\xi_{j}}(\pi^{\xi_{i} \leftrightarrow \xi_{j}}) - c_{\xi_{j}}(\pi(\xi_{j}), \pi(\xi_{i}))
			&= \Delta + \sum_{ \{ \xi_{j}, \eta \} \in \tau ~:~ \{ \pi^{\xi_{i} \leftrightarrow \xi_{j}}(\xi_{j}), \pi^{\xi_{i} \leftrightarrow \xi_{j}}(\eta) \} \in \lambda } w( \{ \xi_{j}, \eta \}) - \frac{\Delta}{2}\\
			&\ge \frac{\Delta}{2} = 5 \cdot \max_{i \in V} |\{ \{ i, j \} \in E \}|\\
			&> 2 \cdot |\{ \xi_{j}, \eta \} \in E \}| \ge  u_{\xi_{j}}(\pi).
	\end{align*}
	On the other hand, when the agent $\xi_{i}$ is allocated the item in $A \setminus \left( A_{i} \cup A_{j} \right)$ by $\pi$, we have that
	\begin{align*}
		u_{\xi_{j}}(\pi^{\xi_{i} \leftrightarrow \xi_{j}}) - c_{\xi_{j}}(\pi(\xi_{j}), \pi(\xi_{i}))
			&= u_{\xi_{j}}(\pi^{\xi_{i} \leftrightarrow \xi_{j}}) + \frac{\Delta}{2}\\
			&\ge \frac{\Delta}{2} = 5 \cdot \max_{i \in V} |\{ \{ i, j \} \in E \}|\\
			&> 2 \cdot |\{ \xi_{j}, \eta \} \in E \}| \ge  u_{\xi_{j}}(\pi).
	\end{align*}
	Hence, the coalitional exchange $\xi_{i} \leftrightarrow \xi_{j}$ is blocking $\pi$. This is a contradiction.
\end{proof}

Next, we prove Property (\ref{Second_Property}): Every permissible coalitional exchange $(X, \mu)$ under a valid allocation $\pi: N \to A$ is a pairwise exchange. Specifically, given a valid allocation $\pi: N \to A$, for each permissible coalitional exchange $(X, \mu)$, there is a coalition $K$ of players in $V$ such that for each $i \in K$, we have two agents $x_{i}$ and $d_{i} \in D_{i}$ satisfying that $(X, \mu)$ coincides with $\{ x_{i} \leftrightarrow d_{i} \}_{i \in K}$. 

\begin{claim}
\label{Claim_Permissible_Exchange}
	Let $\pi: N \to A$ be a valid allocation.
	Every permissible coalitional exchange $(X, \mu)$ under $\pi$ is a pairwise exchange.
	In particular, for each permissible coalitional exchange $(X, \mu)$ under $\pi$, there exists a coalition $K$ of players in $V$ such that for each player $i \in K$, we have two agents $x_{i}$ and $d_{i} \in D_{i}$ satisfying that $(X, \mu) = \{ x_{i} \leftrightarrow d_{i} \}_{i \in K}$.
\end{claim}
\begin{proof}
	By the definition of the cost function, it is easy to see that any exchange such that an agent in $N_{i}$ receives an item in $A \setminus A_{i}$ is not permissible. 
	Hence, for each permissible coalitional exchange $(X, \mu)$ under the vlid allocation $\pi$, there exists a coalition $K$ of players in $V$ such that each plyaer $i \in K$ has a subset $X_{i} \subseteq \{ x_{i} \} \cup D_{i}$, and it holds that $X = \bigcup_{i \in K} X_{i}$ and $\mu(\xi_{i}) \in X_{i} \setminus \{ x_{i} \}$ for each player $i \in K$.
	We replace the permissible coalitional exchange $(X, \mu)$ under $\pi$ with $\{ (X_{i}, \mu_{i}) \}_{i \in K}$, where for each $i \in [k]$, $\mu_{i}: X_{i} \to X_{i}$ is a bijectiosn such that $\mu_{i}(\xi_{i}) \neq \xi_{i}$ for each $\xi_{i} \in X_{i}$.
	We can easily see that the agent $x_{i}$ contains in $X_{i}$ for each player $i \in K$.
	
	For each player $i \in K$, if the coalition $X_{i}$ contains exactly one dummy agent $d_{i} \in D_{i}$, then it satisfies that $c_{x_{i}}(\pi(x_{i}), \pi(d_{i})) + c_{d_{i}}(\pi(d_{i}), \pi(x_{i})) = 0$.
	We now show that for each player $i \in K$, the subset $X_{i}$ does not contain more than one dummy agent in $D_{i}$.
	For the sake of contradiction, we suppose that at least two dummy agents in $D_{i}$ are contained in $X_{i}$. Then, for any bijection $\mu_{i}: X_{i} \to X_{i}$ satisfying that $\mu_{i}(x) \neq x$ for each $x \in X_{i}$, it holds that $\sum_{x \in X_{i}} c_{x}(\pi(x), \pi(\mu(x))) \le 2 - |X_{i}| < 0$. This contradicts that $(X, \mu)$ is permissible.
\end{proof}

Finally, we observe the relationship between a valid allocation for the graphical one-sided matching problem with exchange cost $\mathcal{M}$ and a strategy profile of the network coordination game $\mathcal{G}$. We prove Property (\ref{Third_Property}): Let $\pi: N \to A$ be a valid allocation. We define the strategy profile $s_{\pi} \in \mathbb{S}_{V}$ as $s_{\pi}(i) = f(\pi(x_{i}))$ for each player $i \in V$. Then, for each player $i \in V$, it satisfies that $u_{x_{i}}(\pi) = \Delta + 2 \cdot p_{i}(s_{\pi})$.

\begin{claim}
\label{Claim_Utility}
	Let $\pi: N \to A$ be a valid allocation. We define the strategy profile $s_{\pi} \in \mathbb{S}_{V}$ as $s_{\pi}(i) = f(\pi(x_{i}))$ for each player $i \in V$. Then, for each player $i \in V$, it satisfies that $u_{x_{i}}(\pi) = \Delta + 2 \cdot p_{i}(s_{\pi})$.
\end{claim}
\begin{proof}
	For each player $i \in V$, we have that
	\begin{align*}
		u_{x_{i}}(\pi) 
			&= v_{x_{i}}(\pi(x_{i})) + \sum_{\{ x_{i}, y \} \in \tau ~:~ \{ \pi(x_{i}), \pi(y) \} \in \lambda } w( \{ x_{i}, y \} ) \\
			&= \Delta + \sum_{\{ x_{i}, y \} \in \tau ~:~ f(\pi(x_{i})) = f(\pi(y)) } w( \{ x_{i}, y \} ) \\
			&= \Delta + \sum_{\{ x_{i}, x_{j} \} \in E ~:~ f(\pi(x_{i})) = f(\pi(x_{j})) } w( \{ x_{i}, x_{j} \} ) \\
			&= \Delta + \sum_{\{ x_{i}, x_{j} \} \in E ~:~ s_{\pi}(i) = s_{\pi}(j) } w( \{ x_{i}, x_{j} \} ) \\
			&= \Delta + 2 \cdot \left| \left\{  \{ x_{i}, x_{j} \} \in E ~:~ s_{\pi}(i) = s_{\pi}(j) \right\} \right| \\
			&= \Delta + 2 \cdot p_{i}(s_{\pi}).
	\end{align*}
	Therefore, the claim follows.
\end{proof}

What remains is to prove that the network coordination game $\mathcal{G}$ has a $k$-equilibrium if and only if the graphical one-sided matching problem without exchange cost $\mathcal{M}$ has a $2k$-stable allocation.

We first assume that $\mathcal{M}$ has a $2k$-stable allocation $\pi: N \to A$. Then, we now show that $\mathcal{G}$ has a $k$-equilibrium.
By Claim \ref{Claim_Stable_is_Valid}, the allocation $\pi$ is valid. We define the strategy profile $s_{\pi} = (s_{\pi}(i))_{i \in V} \in \mathbb{S}_{V}$ with respect to $\pi$ by $s_{\pi}(i) = f(\pi(x_{i}))$ for each player $i \in V$.

From now on, we prove that $s_{\pi}$ is a $k$-equilibirum in $\mathcal{G}$.
For the sake of contradiction, we suppose that the strategy profile $s_{\pi}$ is not a $k$-equilibrium. Thus, there is a profitable deviation $s_{\pi} \mapsto_{K} s$ where $K$ is a coalition of players in $V$ with the size $|K| \le k$.
For each player $i \in K$, we denote by $d_{i}$ the dummy agent assigned the item $a_{i, s(i)}$ by the allocation $\pi$, i.e., $d_{i} := \pi^{-1}(a_{i, s(i)})$.
We consider the pairwise exchange $\{ x_{i} \leftrightarrow d_{i} \}_{i \in K}$. It is easy to see that this pairwise exchange is permissible.
Since the deviation $s_{\pi} \mapsto_{K} s$ is profitable, it holds that for each player $i \in K$, $p_{i}(s_{\pi}) < p_{i}(s)$.
Thus, the pairwise exchange $\{ x_{i} \leftrightarrow d_{i} \}_{i \in K}$ is blocking $\pi$ because for each player $i \in K$, it holds that 

\begin{align*}
	u_{x_{i}} \left(\pi^{(X, \mu)} \right) - u_{x_{i}}(\pi) 
		&= \left( \Delta + 2 \cdot p_{i}(s)  \right) - \left( \Delta + 2 \cdot p_{i}(s_{\pi}) \right) \\
		&= 2 \left( p_{i}(s) - p_{i}(s_{\pi}) \right)
		\ge 2 > 1 = c_{x_{i}}\left( \pi(x_{i}), \pi^{(X, \mu)}(x_{i}) \right)
\end{align*}

\noindent
and

\begin{align*}
	u_{d_{i}} \left( \pi^{(X, \mu)} \right) - c_{d_{i}} \left( \pi(d_{i}), \pi^{(X, \mu)}(d_{i}) \right) 
		= \Delta + 1 > \Delta = u_{d_{i}}(\pi).
\end{align*}

\noindent
Note that the size of coalition $\bigcup_{i \in K} \{ x_{i}, d_{i} \}$ is bounded by $2k$. This is a contradiction.
Hence, the strategy profile $s_{\pi}$ is a $k$-equilibrium. The network coordination game $\mathcal{G}$ has a $k$-equilibrium.

Next, we assume that $\mathcal{G}$ has a $k$-equilibrium $s^{*} = (s_{i}^{*})_{i \in V} \in \mathbb{S}_{V}$. Then, we now show that $\mathcal{M}$ has a $2k$-stable allocation. 
Before constructing a corresponding allocation, we introduce a rank function $\mathrm{rank}_{i}: S_{i} \to [m]$ for each player $i \in V$. For each player $i \in V$ and each pure strategy $\sigma_{i} \in S_{i}$, we define $\mathrm{rank}_{i}(\sigma_{i}) := | \{ s \in S_{i} ~;~ s \le \sigma_{i} \}|$.
We construct an allocation $\pi_{s^{*}}: N \to A$ from $s^{*}$ as follows: For each agent $i \in V$, $\pi_{s^{*}}(x_{i}) = a_{i, s_{i}^{*}}$, for each dummy agent $d_{i}^{k} \in D_{i}$ with $k < \mathrm{rank}_{i}(s_{i})$, $\pi_{s^{*}}(d_{i}^{k}) = a_{i, \mathrm{rank}_{i}^{-1}(k)}$, and for each dummy agent $d_{i}^{k} \in D_{i}$ with $k \ge \mathrm{rank}_{i}(s_{i})$, $\pi_{s^{*}}(d_{i}^{k}) = a_{i, \mathrm{rank}_{i}^{-1}(k+1)}$.

From now on, we prove that $\pi_{s^{*}}$ is a $2k$-stable allocation on $\mathcal{M}$.
For the sake of contradiction, we suppose that the allocation $\pi_{s^{*}}$ is not $2k$-stable; thus, there exists a permissible coalitional exchange $(X, \mu)$ with $|X| \le 2k$ that is blocking the allocation $\pi_{s^{*}}$.
From Claim \ref{Claim_Permissible_Exchange}, we have a coalition $K$ of players in $V$ such that for such player $i \in K$, there is a dummy agent $d_{i} \in D_{i}$ such that the pairwise exchange $\{ x_{i} \leftrightarrow d_{i} \}_{i \in K}$ coincides with the coalitional exchange $(X, \mu)$.
We denote by $t_{\pi_{s^{*}}^{(X, \mu)}} = (t_{\pi_{s^{*}}^{(X, \mu)}}(i))_{i \in V} \in \mathbb{S}_{V}$ the strategy profile that is to be a deviation of players in $K$ from the strategy profile $s^{*}$ according to the allocation $\pi_{s^{*}}^{(X, \mu)}$. That is, $t_{\pi_{s^{*}}^{(X, \mu)}}(i) := \pi_{s^{*}}^{(X, \mu)}(x_{i}) \neq \pi_{s^{*}}(x_{i}) = s^{*}(i)$ for each player $i \in K$, and $t_{\pi_{s^{*}}^{(X, \mu)}}(i) := \pi_{s^{*}}^{(X, \mu)}(x_{i}) = \pi_{s^{*}}(x_{i}) = s^{*}(i)$ for each player $i \in V \setminus K$.
Since the coalitional exchange $(X, \mu)$ is blocking $\pi$, for every player $i \in K$, the following inequality holds: \[ u_{x_{i}}\left( \pi^{(X, \mu)} \right) - c_{x_{i}} \left( \pi(x_{i}), \pi^{(X, \mu)} \right) = \Delta + 2 \cdot p_{i} \left( t_{\pi_{s^{*}}^{(X, \mu)}} \right) - 1 > \Delta + 2 \cdot p_{i} \left( s^{*} \right) = u_{x_{i}}(\pi). \] Hence, we have that $p_{i} \left( t_{\pi_{s^{*}}^{(X, \mu)}} \right) > p_{i} \left( s^{*} \right)$. This implies that $s^{*} \mapsto_{K} t_{\pi_{s^{*}}^{(X, \mu)}}$ is a profitable deviation.
Furthermore, the number of players in $K$ is bounded by $k$. Thus, $s^{*}$ is not a $k$-equilibrium. This is a contradiction. Hence, the graphical one-sided matching problem with exchange cost $\mathcal{M}$ has a $2k$-stable allocation.

\section{Proof of Theorem \ref{Theorem_Check_Stable_without_Costs_coNP}}
\label{Appendix__Theorem_Check_Stable_without_Costs_coNP}
In this section, we prove Theorem \ref{Theorem_Check_Stable_without_Costs_coNP}.
That is, we show that the problem \textsc{Check Stable Allocation without Exchange Cost} is $\mathtt{coNP}$-complete.

	It is easy to see that the problem \textsc{Check Stable Allocation without Exchange Cost} belongs to $\mathtt{coNP}$.
	Therefore, the rest of this proof shows the $\mathtt{coNP}$-hardness of the problem \textsc{Check Stable Allocation without Exchange Cost}.
	We now show a polynomial-time reduction from \textsc{Check Equilibrium} to \textsc{Check Stable Allocation without Exchange Cost}. Recall that the problem \textsc{Check Equilibrium} has been proven the $\mathtt{coNP}$-completeness by Apt, de Keijzer, Rhan, Sch\"{a}fer, and Simon \cite{AKRSS17}.
	
	\begin{theorem}[Apt, de Keijzer, Rhan, Sch\"{a}fer, and Simon \cite{AKRSS17}]
		The problem \textsc{Check Equilibrium} is $\mathtt{coNP}$-complete.
	\end{theorem}

	Let $\left\langle \mathcal{G} = \left( G = (V, E), m, (S_{i})_{i \in V} \right), s = (s_{i})_{i \in V}, k \right\rangle$ be an instance of \textsc{Check Equilibrium}.
	We construct a graphical one-sided matching problem without exchange cost $\mathcal{M} = \left\langle \mathcal{A} = (N, \tau, \{ w: \tau \to \mathbb{Z}_{> 0} \}), \right. \allowbreak \left. \mathcal{I} = (A, \lambda), \{ v_{i}: A \to \mathbb{Z}_{\ge 0} \}_{i \in V} \right\rangle$ and an allocation $\pi: N \to A$.
	The construction of the agent graph $\mathcal{A}$ and the item graphs $\mathcal{I}$ is the same as the graph shown in the proof of Lemma \ref{Lemma_Reduction_from_Coordination_Games_to_One-Sided_Market}.
	Due to the different computational problems --- specifically, without exchange costs --- we modify the valuation function $v_{d}: A \to \mathbb{Z}_{\ge 0}$ for each dummy agent $d \in \bigcup_{i \in V} D_{i}$.
	For each player $i \in V$ and for each dummy agent $d_{i} \in D_{i}$, we define the new valuation function $v_{d_{i}}: A \to \mathbb{Z}_{\ge 0}$ as follows: For each item $a \in A$,
	
	\begin{align*}
		v_{d_{i}}(a) = \begin{cases}
			\Delta & \text{ if } a = a_{i, s_{i}},\\
			0 & \text{ otherwise.}
		\end{cases}
	\end{align*}
	
	\noindent
	Finally, we define the allocation $\pi: N \to A$ corresponding to the strategy profile $s = (s_{i})_{i \in V}$ as follows.
	For each player $i \in V$, $\pi(x_{i}) = a_{i, s_{i}}$. 
	For each player $i \in V$ and each dummy agent $d_{i}^{k} \in D_{i}$ with $k < \mathrm{rank}_{i}(s_{i})$, $\pi(d_{i}^{k}) = a_{i, \mathrm{rank}_{i}^{-1}(k)}$.
	For each player $i \in V$ and for each dummy agent $d_{i}^{k} \in D_{i}$ with $k \ge \mathrm{rank}_{i}(s_{i})$, $\pi(d_{i}^{k}) = a_{i, \mathrm{rank}_{i}^{-1}(k+1)}$.
	It is easy to see that the allocation $\pi$ is valid.
	
	What remains is to prove that the strategy profile $s$ is a $k$-equilibrium for the coordination game $\mathcal{G}$ if and only if the allocation $\pi$ is a $2k$-stable allocation for the graphical one-side matching problem $\mathcal{M}$.
	
	First, we assume that the strategy profile $s$ is not a $k$-equilibrium for the coordination game $\mathcal{G}$. Thus, we have a profitable deviation $s \mapsto_{K} s'$ satisfying that $|K| \le k$.
	We denote by $d_{i}$ the dummy agent assigned the item $s'(i)$ by the allocation $\pi$, i.e., $d_{i} := \pi^{-1}(s'(i))$.
	Then, we consider the coalitional exchange $(X, \mu) = \{x_{i} \leftrightarrow d_{i} \}_{i \in K}$.
	We now show that $(X, \mu)$ is blocking $\pi$.
	By definition, the utility of every dummy agent $d_{i}$ in $X$ is strictly increasing: \[ u_{d_{i}}(\pi^{(X, \mu)}) = \Delta > 0 = u_{d_{i}}(\pi). \]
	Since $s \mapsto_{K} s'$ is a profitable deviation, it holds that $p_{i}(s') > p_{i}(s)$ for each player $i$ in $K$. Hence, we have that for each player $i$ in $K$,
	
	\begin{align*}
		u_{x_{i}}(\pi^{(X, \mu)}) = \Delta + 2 \cdot p_{i}(s') > \Delta + 2 \cdot p_{i}(s) = u_{x_{i}}(\pi).
	\end{align*}
	
	\noindent
	Furthermore, it satisfies that $|X| \le 2k$. Therefore, the allocation $\pi$ is not a $2k$-stable allocation for $\mathcal{M}$.
	
	Next, we assume that the allocation $\pi: N \to A$ is not a $2k$-stable allocation for $\mathcal{M}$. Thus, we have a blocking exchange $(X, \mu)$ with the size $|X| \le 2k$.
	We show that there is a subset $K$ of players in $V$ with the size $|K| \le k$ such that $X := \bigcup_{i \in K} \{ x_{i}, d_{i} \}$, $\mu(x_{i}) = d_{i}$, and $\mu(d_{i}) = x_{i}$, where $d_{i} \in D_{i}$ for each player $i \in K$. For each player $i \in V$, the utility of the dummy agent $\hat{d}_{i} \in D_{i}$ is strictly increasing only when $\hat{d}_{i}$ receives the item $a_{i, s_{i}}$. Furthermore, when the agent $x_{i}$ receives any item in $A \setminus A_{i}$, the utility of $x_{i}$ is strictly decreasing.
	
	We denote by $s' = (s'_{i})_{i \in V} \in \mathbb{S}_{V}$ the strategy profile corresponding to $\pi^{(X, \mu)}$, i.e., for each $i \in K$, $s'_{i} := \pi^{(X, \mu)}(i)$ and for each $j \in V \setminus K$, $s'_{j} := \pi^{(X, \mu)}(j) = \pi(j)$.
	From Claim \ref{Claim_Utility}, it holds that \[ u_{x_{i}}(\pi^{(X, \mu)}) = \Delta + 2 \cdot p_{i}(s') > \Delta + 2 \cdot p_{i}(s) = u_{x_{i}}(\pi). \]
	This implies that $s \mapsto_{K} s'$ is a profitable deviation. Since $|K| \le k$, the strategy profile $s$ is not a $k$-equilibrium.

\section{Proof of Theorem \ref{Theorem_Existence_Two-Stable}}
\label{Appendix__Theorem_Existence_Two-Stable}
In this section, we prove Theorem \ref{Theorem_Existence_Two-Stable}.
That is, we show that every graphical one-sided matching with exchange cost has at least one $2$-stable allocation.

	Let $\mathcal{M} = \langle \mathcal{A} = (N, \tau, \{ w: \tau \to \mathbb{R}_{> 0} \}), \mathcal{I} = (A, \lambda), \{ v_{i}: A \to \mathbb{R}_{\ge 0} \}_{i \in N}, \{ c_{i}: A \times A \to \mathbb{R} \}_{i \in N} \rangle$ be a graphical one-sided matching problem with exchange cost.
	We consider the same function $\phi: \Pi \to \mathbb{R}_{\ge 0}$ as that used in \cite{MS19}: For each allocation $\pi: N \to A$,
	
	\begin{align*}
		\phi(\pi) := \sum_{i \in N} \{ v_{i}(\pi(i)) + u_{i}(\pi) \}.
	\end{align*}
	
	\noindent
	We now prove that the function $\phi$ acts as a generalized ordinal potential function even where the cost functions exist.
	Thus, we show that for every allocation $\pi: N \to A$ and each permissible coalitional exchange $i \leftrightarrow j$, if $u_{i}(\pi^{i \leftrightarrow j}) - c_{i} \left( \pi(i), \pi^{i \leftrightarrow j}(i) \right) > u_{i}(\pi)$ and $u_{j}(\pi^{i \leftrightarrow j}) - c_{j} \left( \pi(j), \pi^{i \leftrightarrow j}(j) \right) > u_{j}(\pi)$ hold, then it satisfies that $\phi(\pi^{i \leftrightarrow j}) > \phi(\pi)$.
	
	We observe the key property of the potential function $\phi$ in Lemma \ref{Lemma_KeyProperty_of_Potential}.
	To prove this, we introduce some additional notations.
	Given an allocation $\pi: N \to A$, we define $\tau_{\pi}$ as $\{ \{ i, j \} \in \tau ~;~ \{ \pi(i), \pi(j) \} \in \lambda \}$.
	For any coalition $X$ of agents in $N$, we define $\tau[X] := \{ \{ x, y \} \in \tau ~;~ x, y \in X \}$ and $\hat{\tau}[X] := \{ \{ x, i \} \in \tau ~;~ x \in X \text{ and } i \in N \setminus X \}$.
	
	\begin{lemma}
	\label{Lemma_KeyProperty_of_Potential}
		Let $\mathcal{M} = \langle \mathcal{A} = (N, \tau, \{ w: \tau \to \mathbb{R}_{> 0} \}), \mathcal{I} = (A, \lambda), \{ v_{i}: A \to \mathbb{R}_{\ge 0} \}_{i \in N}, \{ c_{i}: A \times A \to \mathbb{R} \}_{i \in N} \rangle$ be a graphical one-sided matching problem with exchange cost.
		For every coalitional exchange $(X, \mu)$ (not necessarily permissible), it holds that 
		\begin{align*}
			&\phi(\pi^{(X, \mu)}) - \phi(\pi) \\ &= 2 \left( \sum_{x \in X} u_{x}(\pi^{(X, \mu)}) - \sum_{x \in X} u_{x}(\pi) - \sum_{ \{ x, y \} \in \tau_{\pi^{(X, \mu)}} \cap \tau[X] } w( \{ x, y \} ) + \sum_{ \{ x, y \} \in \tau_{\pi} \cap \tau[X] } w( \{ x, y \}) \right).
		\end{align*}
	\end{lemma}
	\begin{proof}
		\begin{align*}
			&\phi(\pi^{(X, \mu)}) - \phi(\pi) \\
				&= \sum_{i \in N} \left\{ v_{i}(\pi^{(X, \mu)}(i)) + u_{i}(\pi^{(X, \mu)}) \right\} - \sum_{i \in N} \left\{ v_{i}(\pi(i)) + u_{i}(\pi) \right\} \\
				&= 2 \left( \sum_{x \in X} v_{x}(\pi^{(X, \mu)}) - \sum_{x \in X} v_{x}(\pi) \right) + \sum_{i \in N} \sum_{ \{i, j\} \in \tau_{\pi^{(X, \mu)}}} w(\{i, j\}) - \sum_{i \in N} \sum_{ \{i, j\} \in \tau_{\pi}} w(\{i, j\}) \\
				&= 2 \left( \sum_{x \in X} v_{x}(\pi^{(X, \mu)}) - \sum_{x \in X} v_{x}(\pi) \right) \\
				&~~~~~~~~~	+ \sum_{i \in N \setminus X} \sum_{ \{i, j\} \in \tau_{\pi^{(X, \mu)}} \cap \hat{\tau}[X]} w(\{i, j\}) - \sum_{i \in N \setminus X} \sum_{ \{i, j\} \in \tau_{\pi} \cap \hat{\tau}[X]} w(\{i, j\}) \\
				&~~~~~~~~~	+ \sum_{x \in  X} \sum_{ \{x, j\} \in \tau_{\pi^{(X, \mu)}} \cap \hat{\tau}[X]} w(\{x, j\}) - \sum_{x \in X} \sum_{ \{x, j\} \in \tau_{\pi} \cap \hat{\tau}[X]} w(\{x, j\}) \\
				&~~~~~~~~~	+ \sum_{x \in  X} \sum_{ \{x, y\} \in \tau_{\pi^{(X, \mu)}} \cap \tau[X]} w(\{x, y\}) - \sum_{x \in X} \sum_{ \{x, j\} \in \tau_{\pi} \cap \tau[X]} w(\{x, y\}) \\
				&= 2 \left( \sum_{x \in X} v_{x}(\pi^{(X, \mu)}) - \sum_{x \in X} v_{x}(\pi) \right) \\
				&~~~~~~~~~	+ 2 \left( \sum_{x \in  X} \sum_{ \{x, j\} \in \tau_{\pi^{(X, \mu)}} \cap \hat{\tau}[X]} w(\{x, j\}) - \sum_{x \in X} \sum_{ \{x, j\} \in \tau_{\pi} \cap \hat{\tau}[X]} w(\{x, j\}) \right) \\
				&~~~~~~~~~	+ \sum_{x \in  X} \sum_{ \{x, y\} \in \tau_{\pi^{(X, \mu)}} \cap \tau[X]} w(\{x, y\}) - \sum_{x \in X} \sum_{ \{x, j\} \in \tau_{\pi} \cap \tau[X]} w(\{x, y\}) \\
				&= 2 \left( \sum_{x \in X} \left\{ v_{x}(\pi^{(X, \mu)}) + \sum_{ \{ x, y \} \in \tau_{\pi^{(X, \mu)}}} w(\{x, y\}) \right\} - \sum_{x \in X} \left\{ v_{x}(\pi) + \sum_{ \{ x, y \} \in \tau_{\pi}} w(\{x, y\}) \right\} \right) \\
				&~~~~~~~~~	- \sum_{x \in  X} \sum_{ \{x, y\} \in \tau_{\pi^{(X, \mu)}} \cap \tau[X]} w(\{x, y\}) + \sum_{x \in X} \sum_{ \{x, j\} \in \tau_{\pi} \cap \tau[X]} w(\{x, y\}) \\
				&= 2 \left( \sum_{x \in X} u_{x}(\pi^{(X, \mu)}) - \sum_{x \in X} u_{x}(\pi) - \sum_{ \{ x, y \} \in \tau_{\pi^{(X, \mu)}} \cap \tau[X] } w( \{ x, y \} ) + \sum_{ \{ x, y \} \in \tau_{\pi} \cap \tau[X] } w( \{ x, y \} ) \right).
		\end{align*}
		Therefore, this lemma follows.
	\end{proof}
	
	Now, we fix an allocation $\pi: N \to A$, and we assume that there is a permissible coalitional exchange $i \leftrightarrow j$ that is blocking $\pi$. Hence, the following three inequalities hold: $u_{i}(\pi^{i \leftrightarrow j}) - c_{i} \left( \pi(i), \pi^{i \leftrightarrow j}(i) \right) > u_{i}(\pi)$, $u_{j}(\pi^{i \leftrightarrow j}) - c_{j} \left( \pi(j), \pi^{i \leftrightarrow j}(j) \right) > u_{j}(\pi)$, and $c_{i} \left( \pi(i), \pi^{i \leftrightarrow j}(i) \right) + c_{j} \left( \pi(j), \pi^{i \leftrightarrow j}(j) \right) \ge 0$.
	Therefore, we have that \[ \left( u_{i}(\pi^{i \leftrightarrow j}) + u_{j}(\pi^{i \leftrightarrow j}) \right) - \left( u_{i}(\pi) + u_{j}(\pi) \right) > c_{i} \left( \pi(i), \pi^{i \leftrightarrow j}(i) \right) + c_{j} \left( \pi(j), \pi^{i \leftrightarrow j}(j) \right) \ge 0. \]
	Furthermore, every coalitional exchange $i' \leftrightarrow j'$, it is easy to see that $\{ i', j' \} \in \tau_{\pi}$ if and only if $\{ i', j' \} \in \tau_{\pi^{i' \leftrightarrow j'}}$.
	This implies that \[ \sum_{ \{ x, y \} \in \tau_{\pi} \cap \tau[\{i, j \}] } w( \{ x, y \}) = \sum_{ \{ x, y \} \in \tau_{\pi^{i \leftrightarrow j}} \cap \tau[\{ i, j\}] } w( \{ x, y \} ). \]
	From the above observation, it holds that $\phi(\pi) < \phi(\pi^{i \leftrightarrow j})$ if an allocation $\pi: N \to A$ has a permissible coalitional exchange $i \leftrightarrow j$ that is blocking $\pi$.
	Therefore, the function $\phi$ acts as a generalized ordinal potential function.

	Let $\pi: N \to A$ be an allocation satisfying that $\phi(\pi) \ge \phi(\pi^{i \leftrightarrow j})$ for every permissible exchange $i \leftrightarrow j$.
	Since the function $\phi$ is a generalized ordinal potential function, it satisfies at least one of $u_{i}(\pi^{i \leftrightarrow j}) - c_{i} \left( \pi(i), \pi^{i \leftrightarrow j}(i) \right) \le u_{i}(\pi)$ and $u_{j}(\pi^{i \leftrightarrow j}) - c_{j} \left( \pi(j), \pi^{i \leftrightarrow j}(j) \right) \le u_{j}(\pi)$ for each permissible exchange $i \leftrightarrow j$. This implies that the allocation $\phi$ is $2$-stable. 
	Note that the number of allocations is finite. Hence, such an allocation always exists.

\section{Proof of Lemma \ref{Lemma_Find_Two_Stable_in_PLS}}
\label{Appendix__Lemma_Find_Two_Stable_in_PLS}
In this section, we prove Lemma \ref{Lemma_Find_Two_Stable_in_PLS}.
That is, we show that the problem of finding a $2$-stable allocation of a graphical one-sided matching problem with exchange cost belongs to $\mathtt{PLS}$.
To prove this, we show a polynomial-time reduction from \textsc{Find $2$-Stable Allocation} to \textsc{Basic PLS Problem}.

Let $\mathcal{M}$ be an instance of \textsc{Find $2$-Stable Allocation}. Here, we denote by $N$ the agent set of $\mathcal{M}$.
For each allocation $\pi \in \Pi$, the neighborhood structure $\mathcal{N}(\pi)$ is defined as $\{ \pi^{i \leftrightarrow j} ~;~$ for all (unordered) permissible pairs $(i, j) \in N \times N \}$.	
Since the size of the neighborhood structure $\mathcal{N}(\pi)$ is $O(|N|^{2})$, we can verify whether a blocking pair $i \leftrightarrow j$ in $\pi$ exists in polynomial time.	
From the proof of Theorem \ref{Theorem_Existence_Two-Stable}, the local maximum point of $\phi$ is a $2$-stable allocation.

\section{Proof of Lemma \ref{Lemma_PLS-Hardness_for_Five-Degree}}
\label{Appendix__Lemma_PLS-Hardness_for_Five-Degree}
In this section, we prove that Lemma \ref{Lemma_PLS-Hardness_for_Five-Degree}.
That is, we show the $\mathtt{PLS}$-hardness of the problem \textsc{Find $2$-Stable Allocation} even if the maximum degree of the underlying agent graph is at most five.

	To prove this, we now show a polynomial-time reduction from the problem \textsc{Local Max-Cut} to the problem \textsc{Find $2$-Stable Allocation}.
	The problem \textsc{Local Max-Cut} is a total search problem: Given an edge-weighted undirected graph $\mathcal{G} = (V, E, \{ \omega: E \to \mathbb{Z}_{> 0} \})$, find an assignment $x \in \{ 0, 1 \}^{V}$ such that for each vertex $v \in V$, \[ \sum_{ \{ v, u \} \in E ~:~ x(v) \neq x(u) } \omega(\{ v, u \}) \ge \sum_{ \{ v, u \} \in E ~:~ x(v) = x(u) } \omega(\{u, v\}). \]
	The problem \textsc{Local Max-Cut} is one of the well-known $\mathtt{PLS}$-complete problems. Specifically, Els\"{a}sser and Tscheuschner \cite{ET11} have shown that the problem \textsc{Local Max-Cut} is $\mathtt{PLS}$-complete even if the maximum degree of the underlying graph is at most five.
	
	\begin{theorem}[Els\"{a}sser and Tscheuschner \cite{ET11}]
	\label{Theorem_MaxcutGame}
		The problem \textsc{Local Max-Cut} is $\mathtt{PLS}$-complete even if the maximum degree of the underlying graph is at most five.
	\end{theorem}
	
	Our reduction is inspired by the corresponding proof by Massand and Simon \cite{MS19}.
	We also use a complete bipartite graph in which both sides have the same number of nodes with the number of agents as the item graph. Due to prove the stronger hardness results, we modify the agent graph.

	Let $\mathcal{G} = (V, E, \{ \omega: E \to \mathbb{Z}_{> 0} \})$ be a \textsc{Local Max-Cut} instance.
	We construct an instance $\mathcal{M}$ of a graphical one-sided matching problem with exchange costs from $\mathcal{G}$.
	First, we describe the structure of the item graph $\mathcal{I} = (A, \lambda)$. The set of items $A$ is a union of the sets of $|V|$ zero-labeled items $A_{0}$ and $|V|$ one-labeled items $A_{1}$, i.e., $A = A_{0} \cup A_{1}$. The set of items on the item graph is $\lambda = \left\{ \{ a_{0}, a_{1} \} ~;~ \forall a_{0} \in A_{0}, \forall a_{1} \in A_{1} \right\}$. Thus, the item graph $\mathcal{I}$ is a complete bipartite graph between zero-labeled items and one-labeled items.
	Second, we describe the structure of the agent graph $\mathcal{A} = (N, \tau, \{ w: \tau \to \mathbb{Z}_{> 0} \})$. The set of agents $N$ is a union of the original vertex set $V$ and the set of $|V|$ dummy agents $D = \{ d_{1}, d_{2}, \dots, d_{|V|} \}$, i.e., $N = V \cup D$. The structure of the agent graph $\mathcal{A}$ is almost the same as the original graph $G$. The set of edges $\tau$ is defined as $E$, i.e., every dummy agent is isolated. The edge-weight function $w: \tau \to \mathbb{Z}_{> 0}$ is defined by $w(e) = 2 \cdot \omega(e)$ for each edge $e \in E$.
	Finally, we define the valuation function $v_{i}: A \to \mathbb{Z}_{\ge 0}$ and the cost function $c_{i}: A \times A \to \mathbb{Z}$ for each agent $i$ in $N$.
	Each agent $i_{v}$ in $V$ has the valuation function $v_{i_{v}}(a) = 0$ for every item $a$ in $A$ and the cost function $c_{i_{v}}$ such that for each pair of two items $a$ and $b$ in $A$,
	
	\begin{align*}
		c_{i_{v}}(a, b) = \begin{cases}
			0 & \text{if both $a$ and $b$ are in $A_{0}$,}\\
			0 & \text{if both $a$ and $b$ are in $A_{1}$,}\\
			1 & \text{otherwise.}
		\end{cases}
	\end{align*}
	
	\noindent
	Each agent $i_{d}$ in $D$ has the valuation function $v_{i_{d}}(a) = 0$ for every item $a$ in $A$ and the cost function $c_{i_{d}}$ such that for each pair of two items $a$ and $b$ in $A$,
	
	\begin{align*}
		c_{i_d}(a, b) = \begin{cases}
			0 & \text{if both $a$ and $b$ are in $A_{0}$,}\\
			0 & \text{if both $a$ and $b$ are in $A_{1}$,}\\
			-1 & \text{otherwise.}
		\end{cases}
	\end{align*}
	
	\noindent
	We complete constructing the graphical one-sided matching problem with exchange cost $\mathcal{M}$ corresponding to the given \textsc{Local Max-Cut} instance $\mathcal{G}$.
	It is easy to see that our reduction is polynomial-time computable.

	What remains is to prove that we can efficiently recover a solution of $\mathcal{G}$ from a $2$-stable allocation on $\mathcal{M}$.
	Let $\pi: N \to A$ be a $2$-stable allocation for $\mathcal{M}$.
	Then, the corresponding assignment $x_{\pi} \in \{0, 1\}^{V}$ is defined to be that for each vertex $i \in V$,
	
	\begin{align*}
		x_{\pi}(i) = \begin{cases}
			0 & \text{if $\pi(i) \in A_{0}$,}\\
			1 & \text{if $\pi(i) \in A_{1}$.}
		\end{cases}
	\end{align*}
	
	We now prove that $x_{\pi}$ is a solution of $\mathcal{G}$.
	For the sake of contradiction, we assume that the assignment $x_{\pi}$ is not a solution of $\mathcal{G}$.
	Thus, there exists a vertex $i \in V$ such that
	\begin{align*}
		\sum_{ \{i, j\} \in E ~:~ x_{\pi}(i) = x_{\pi}(j)} \omega(\{i, j\}) > \sum_{ \{i, j\} \in E ~:~ x_{\pi}(i) \neq x_{\pi}(j)} \omega(\{i, j\}).
	\end{align*}
	Since each weight on edge in $E$ is a positive integer, this inequality implies that
	\begin{align*}
		\sum_{ \{i, j\} \in E ~:~ x_{\pi}(i) = x_{\pi}(j)} \omega(\{i, j\}) - \sum_{ \{i, j\} \in E ~:~ x_{\pi}(i) \neq x_{\pi}(j)} \omega(\{i, j\}) \ge 1.
	\end{align*}
	
	Without loss of generality, we suppose that $x_{\pi}(i) = 0$.	
	By definition, the agent $i$ is assigned the item $\pi(i)$ in $A_{0}$.
	Note that the allocation $\pi$ is a bijection. From the pigeonhole principle, there must exist a dummy agent $d \in D$ assigned the item $\pi(d)$ in $A_{1}$ by $\pi$.
	Furthermore, the pair $i$ and $d$ is permissible since $c_{i}(\pi(i), \pi(d)) + c_{d}(\pi(d), \pi(i)) = 0$.
	The following two inequalities hold:
	
	\begin{align*}
		u_{d}(\pi^{i \leftrightarrow d}) - c_{d}(\pi(d), \pi^{i \leftrightarrow d}(d)) = 1 > 0 = u_{d}(\pi);
	\end{align*}
	
	\noindent
	and
	
	\begin{align*}
		u_{i}(\pi^{i \leftrightarrow d}) - c_{i}(\pi(i), \pi^{i \leftrightarrow d}(i))
			&= \sum_{ e = \{j, i\} \in \tau ~:~ \{ \pi^{i \leftrightarrow d}(j), \pi^{i \leftrightarrow d}(i) \} \in \lambda } w(e) - 1\\
			&= \sum_{ e = \{j, i\} \in \tau ~:~ \{ \pi^{i \leftrightarrow d}(j), \pi^{i \leftrightarrow d}(i) \} \in \lambda } 2 \cdot \omega(e) - 1\\
			&= 2 \cdot \left\{ \sum_{ e = \{j, i\} \in \tau ~:~ \{ \pi^{i \leftrightarrow d}(j), \pi^{i \leftrightarrow d}(i) \} \in \lambda } \omega(e) - 1 \right\} + 1\\
			&\ge 2 \cdot \sum_{ e = \{j, i\} \in \tau ~:~ \{ \pi(j), \pi(i) \} \in \lambda } \omega(e) + 1\\
			&> \sum_{ e = \{j, i\} \in \tau ~:~ \{ \pi(j), \pi(i) \} \in \lambda } 2 \cdot \omega(e)\\
			&= \sum_{ e = \{j, i\} \in \tau ~:~ \{ \pi(j), \pi(i) \} \in \lambda } w(e)
			= u_{v}(\pi)
	\end{align*}
	
	\noindent
	This contradicts our assumption that $\pi$ is a $2$-stable allocation.
	Hence, the assignment $x_{\pi}$ is a solution of the \textsc{Local Max-Cut} instance $\mathcal{G}$.
	We complete the proof of this theorem from Theorem \ref{Theorem_MaxcutGame}.

\begin{remark}
	Massand and Simon \cite{MS19} have proven the $\mathtt{PLS}$-hardness of a graphical one-sided matching problem by a polynomial-time reduction from \textsc{Local Max-Cut}. Remark that their proof can be viewed as a polynomial-time reduction from a \textsc{Uniform Balanced Local Max-Cut}, a variant of \textsc{Local Max-Cut} in which we require any solution additionally holds that the numbers of vertices assigned $0$ and $1$ are the same.
	Formally, the problem \textsc{Uniform Balanced Local Max-Cut} is a total search problem: Given an edge-weighted undirected graph $\mathcal{G} = (V, E, \{ \omega: E \to \mathbb{Z}_{> 0} \})$, find an assignment $x \in \{ 0, 1 \}^{V}$ such that for each vertex $v \in V$, $\sum_{ \{ v, u \} \in E ~:~ x(v) \neq x(u) } \omega(\{ v, u \}) \ge \sum_{ \{ v, u \} \in E ~:~ x(v) = x(u) } \omega(\{u, v\})$, and it satisfies that $|\{ v \in V ~;~ x(v) = 1 \}| = |\{ v \in V ~;~ x(v) = 0 \}|$.
	It is well known that Uniform Balanced Local Max-Cut is PLS-complete even if the maximum degree of the underlying graph is at most six \cite{ET11, SY91}.
\end{remark}

\section{Proof of Theorem \ref{Theorem_Polynomial-time_Computable}}
\label{Appendix__Theorem_Polynomial-time_Computable}
In this section, we prove Theorem \ref{Theorem_Polynomial-time_Computable}. Thus, we show that we have a polynomial-time algorithm for computing a core-stable allocation of a non-graphical one-sided matching problem with exchange cost.

	It is straightforward to see that a non-graphical one-sided matching problem with exchange cost is polynomially equivalent\footnote{For any two computational problems $A$ and $B$, we say that $A$ is polynomially equivalent to $B$ if $A$ and $B$ are polynomial-time reducible to each other.} to a graphical one-sided matching problem with exchange cost that satisfies at least one of the following two properties:
	\begin{enumerate}
		\item The underlying agent graph is either edge-empty or a uniform-weighted complete graph; and
		\item the underlying item graph is either edge-empty or complete.
	\end{enumerate}
	Therefore, in this proof, we show that we have a polynomial-time algorithm to find a core-stable allocation on a non-graphical one-sided matching problem with exchange cost.
	
	Let $\mathcal{M} = \left\langle N, A, \{v_{i}: A \to \mathbb{Z}_{\ge 0} \}_{i \in N}, \{ c_{i}: A \times A \to \mathbb{Z} \}_{i \in N} \right\rangle$ be an instance of a non-graphical one-sided matching problem with exchange cost.
	We construct an instance of maximum weight matching in a bipartite graph.
	We consider the complete bipartite graph $G = (N \cup A, E)$, i.e., for each $(i, a) \in N \times A$, $\{i, a\} \in E$.
	The edge-weight function $w: E \to \mathbb{Z}_{> 0}$ is defined as $w(i, a) = v_{i}(a) + 1$ for all $(i, a) \in N \times A$.
	
	Now, we prove that a maximum weight perfect matching on the graph $\left( G, w \right)$ is a core-stable allocation for the original instance $\mathcal{M}$.
	Let $x \in \{0, 1\}^{N \times A}$ be a maximum weight perfect matching on $\left( G, w \right)$, that is, the following properties holds: For all $i \in N$, $\sum_{a \in A} x(i, a) = 1$ and for all $a \in A$, $\sum_{i \in N} x(i, a) = 1$.
	We define the allocation $\pi_{x}: N \to A$ associated with $x$ as $\pi_{x}(i) = a$ if and only if $x(i, a) = 1$.
	
	For the sake of contradiction, we assume that there is a permissible coalitional exchange $(X, \mu)$ that is blocking $\pi_{x}$. Thus, it satisfies that $v_{\xi} \left( \pi_{x}^{(X, \mu)}(\xi) \right) - c_{\xi} \left( \pi_{x}(\xi), \pi_{x}^{(X, \mu)}(\xi) \right) > v_{\xi}(\pi_{x}(\xi))$ for each agent $\xi \in X$.
	Hence, we have that
	
	\begin{align*}
		\sum_{p \in N} v_{p}(\pi_{x}^{(X, \mu)}(p)) - \sum_{p \in N} v_{p}(\pi_{x}(p))
			&= \sum_{\xi \in X} \left\{ v_{\xi}(\pi_{\xi}^{(X, \mu)}(\xi)) - v_{\xi}(\pi_{x}(\xi)) \right\} \\
			&>  \sum_{\xi \in X} c_{\xi}(\pi_{x}(\xi), \pi_{x}^{(X, \mu)}(\xi)) \\
			&\ge 0,
	\end{align*}
	
	\noindent
	where the final inequality follows from that $(X, \mu)$ is permissible.
	
	Now, we define another perfect matching $y \in \{0, 1\}^{N \times A}$ as $y(p, a) = 1$ if and only if $\pi_{x}'(p) = a$.
	Then, we have that
	
	\begin{align*}
		\sum_{(p, a) \in N \times A} w(p, a) \cdot x(p, a) &- \sum_{(p, a) \in N \times A} w(p, a) \cdot y(p, a)\\
			&= \sum_{(p, a) \in N \times A} \left( v_{p}(a) + 1 \right) \cdot x(p, a) - \sum_{(p, a) \in N \times A} \left( v_{p}(a) + 1 \right) \cdot y(p, a)\\
			&= \sum_{p \in N} v_{p}(\pi_{x}'(p)) - \sum_{p \in N} v_{p}(\pi_{x}(p))\\
			&> 0.
	\end{align*}	
	
	\noindent
	This implies that $x$ is not a maximum weight perfect matching, which is a contradiction.
	
	Finally, we can solve the maximum weight perfect matching in complete bipartite matching in polynomial time \cite{DS12}. Hence, we can also obtain a core-stable allocation for $\mathcal{M}$.

\end{document}